\documentclass{article}
\usepackage{array}
\newcolumntype{C}{>{\centering\arraybackslash}p{5cm}}
\usepackage{times}
\usepackage{graphicx}
\usepackage{color}
\usepackage[margin=1in]{geometry}
\usepackage{adjustbox}
\usepackage{amsmath}
\usepackage{amssymb}
\usepackage[table]{xcolor}
\usepackage{mathtools}
\usepackage{amsthm}
\newcommand{\withnotes}{1}
\newcommand{\withcolors}{1}
\usepackage{ccanonne}
\allowdisplaybreaks

\def\<{\langle}
\def\>{\rangle}
\def\be{\begin{equation*}}
\def\ee{\end{equation*}}
\def\bea{\begin{eqnarray*}}
\def\eea{\end{eqnarray*}}

\def\C{\mathbb{C}}

\def\eps{\varepsilon}

\newcommand{\tr}{\operatorname{Tr}}

\newcommand{\dotprodB}[2]{\dotprod{#1}{#2}_B}
\newcommand{\dotprodH}[2]{\dotprod{#1}{#2}_\mathcal{D}}
\newcommand{\parity}[1]{\chi_S\left(#1\right)}

\newcommand{\nsamp}{n}
\newcommand{\ket}[1]{\ensuremath{\left|#1\right\rangle}}
\newcommand{\op}[2]{|#1\rangle \langle #2|}

\newcommand{\krawt}[2]{K^{(N)}_{#1}\left(#2\right)}

\definecolor{pastelgreen}{RGB}{225, 242, 225}

\newcommand{\dims}{d}

\newcommand{\unif}{{\mathbf{u}}}

\newcommand{\povmset}{{\mathfrak{M}}}
\newcommand{\nqubits}{{N}}

\newcommand{\opnorm}[1]{{\left\|#1\right\|}_{\text{op}}}
\newcommand{\tracenorm}[1]{{\left\|#1\right\|}_{1}}
\newcommand{\hsnorm}[1]{{\left\|#1\right\|}_{\text{HS}}}
\newcommand{\barDelta}{{\overline{\Delta}}}

\newcommand{\x}{\mathbf{x}}
\newcommand{\out}{{x}}

\def\multiset#1#2{\ensuremath{\left(\kern-.3em\left(\genfrac{}{}{0pt}{}{#1}{#2}\right)\kern-.3em\right)}}

\newcommand{\qmm}{{\rho_{\text{mm}}}}

\newcommand{\Var}{\text{Var}}
\newcommand{\eye}{\mathbb{I}}

\newcommand{\Luders}{\mathcal{L}}

\newcommand{\POVM}{\mathcal{M}}

\title{Near-Optimal Mixedness Testing with Pauli Measurements}
\author{
    \begin{tabular}[t]{C@{\extracolsep{6.5em}} C}
   Jayadev Acharya &Abhilash Dharmavarapu \\
 Cornell University & Cornell University\\ 
\small \texttt{acharya@cornell.edu} &\small \texttt{ad2255@cornell.edu} 
\end{tabular}
\vspace{2ex}\\
\begin{tabular}[t]{C@{\extracolsep{6.5em}} C}
    Yuhan Liu & Nengkun Yu \\
Rice University & Stony Brook University\\ 
\small \texttt{yuhan-liu@rice.edu} &\small \texttt{nengkun.yu@cs.stonybrook.edu} 
\end{tabular}}

\begin{document}

\maketitle

% \begin{abstract}
% We provide a sublinear algorithm for quantum state certification using Pauli measurements, where one is given
% $n$ copies of an unknown $d$-dimensional quantum state $\sigma$ of rank at most $r$, and
% one wants to test whether $\sigma = \sigma$ for some known
% rank-$r$ state $\sigma$, or whether $\sigma$ is $\varepsilon$-far from $\sigma$. The sample complexity of our algorithm is
% $\tildeO{\frac{r\cdot (3^{p*}/\sqrt{2})^{\nqubits}}{\eps^{2}}}
% \approx \tildeO{\frac{r\cdot 1.88^{\,\nqubits}}{\varepsilon^{2}}}=o(\frac{r d}{\varepsilon^{2}})$,
% where $p* \in(1/2,1)$ satisfies $H_2(p*)=1/2$. Along with our upper bound, we prove a lower bound $\tildeOmega{\frac{\sqrt{10}^{\,\nqubits}}{\varepsilon^{2}}}
% \approx
% \tildeOmega{\frac{3.16^{\,\nqubits}}{\varepsilon^{2}}} \mathscr{\Omega}$ for rank-$d$ states.
% \end{abstract}
\begin{abstract}
     We consider a fundamental problem of \emph{mixedness testing}: Given $n$ copies of an $N$-qubit state $\rho$, determine whether $\rho = \mathbb{I}_d/d$ or $\|\rho-\mathbb{I}_d/d\|_1 \geq \varepsilon$ with high probability, where $d = 2^N$. In particular, we focus on performing this task with simple single-qubit measurements, i.e. measurements that are prepared independently on each qubit. We provide a nearly complete picture of single-qubit mixedness testing by showing that its copy-complexity is $\Tilde{\Theta}\left(\sqrt{10}^N/\varepsilon^2\right)$. To establish our lower bound, we introduce a new measurement-dependent lower bound framework for adaptive single-copy state certification. For the upper bound, we present a randomized Pauli basis measurement protocol, which relies on a new primitive for computationally efficient uniformity testing of correlation-concentrated distributions on the Boolean hypercube. In conjunction, we provide lower and upper bounds for mixedness testing with fixed Pauli measurement protocols.
\end{abstract}

\section{Introduction}
The goal of \emph{state certification} is to determine whether a state is $\eps$-far or equal to a known target state $\sigma$. 
More formally, we perform measurements on $n$ copies of a $N$-qubit mixed state $\rho$ to determine whether $\|\rho-\sigma\|_1 \geq \eps$ or $\rho=\sigma$ with high probability. The problem poses important applications in quantum computing, such as evaluating state degradation in quantum devices and circuits and verifying outcomes of quantum experiments. 
An approach to certification is to compute the density matrix $\rho$ directly using quantum state tomography (QST). The issue is that this procedure can be costly, even in the most powerful measurement setting. 
On the other hand, we can perform state certification by instead learning simpler properties of the state, which uses exponentially fewer number of copies than QST. 

A scenario of special interest is when the target state is the maximally mixed state, i.e. $\sigma = \eye_d/d = \qmm$. 
We often refer to this task as \emph{mixedness testing}. The copy complexity of mixedness testing has been well studied under entangled measurements, where measurements are performed simultaneously on $n$ copies, and single-copy measurements, where measurements are performed sequentially on each copy. For entangled measurements, it has been shown that $n = \bigTheta{2^\nqubits/\eps^2}$\cite{BOW17, BadescuO019}. Under single-copy measurements, copy-complexity has been resolved to be $n = \bigTheta{\sqrt{8}^\nqubits/\eps^2}$ ~\cite{BubeckC020, Chen0HL22}. In particular, the results hold even when measurements are adaptively chosen from the previous measurement outcomes.

In practice, preparing the aforementioned measurements can often be difficult with large qubit systems, as they require entangling $\nqubits$ qubits. For example, the upper bound for single-copy mixedness is attained by performing a Haar measurement and classically performing uniformity testing on the outcome distribution~\cite{BubeckC020}. 
Even using approximate unitary designs to simulate the Haar measure requires repeated use of multi-qubit gates, which may be prone to error given the current quality of quantum hardware. 
It is often easier to prepare measurements independently on each qubit, which we refer to as single-qubit measurements. However, relatively little is known about the copy complexity of single-qubit mixedness testing.

We nearly settle the question of the copy-complexity of this problem, showing that $n = \tildeTheta{\sqrt{10}^\nqubits/\eps^2}$ copies are sufficient and necessary for single-qubit certification with copy-wise adaptivity, where the upper bound is attained by a non-adaptive scheme. We improve the previously best known upper bound of $\tildeO{4^\nqubits/\eps^2}$ \cite{Yu2023almost}, which was using more restrictive Pauli observable measurements. 
Prior to our work, there were no known results for the lower bound other than $\bigOmega{\sqrt{8}^N/\eps^2}$ under single-copy measurements.

In addition, shared randomness can be also a valuable resource. We further our study into the copy-complexity of deterministic single-qubit protocols, specifically Pauli measurements. We provide new lower and upper bounds for the copy-complexity under this measurement setting.
\subsection{Problem Setup}
We formalize the problem of mixedness testing. The goal is to design a measurement scheme $\POVM^n= \POVM^{(1)} \otimes \POVM^{(2)} \otimes \ldots \otimes \POVM^{(\nsamp)}$ such that 
\begin{align*}
    \Pr_\qmm[T(x^n) = \mathtt{Maximally-Mixed}] \geq 1-\delta, \quad
    \min_{\sigma: \tracenorm{\sigma-\qmm} \geq \eps}\Pr_{\sigma}[T(x^n) = \mathtt{Not \; Maximally-Mixed}] \geq 1-\delta,
\end{align*}
where $x^n$ are the outcomes of the measurement $\POVM^n$ on $\sigma^{\otimes n}$. Furthermore, we impose the restriction that $\forall_{i \in [n]} \POVM_i \in \povmset$. Throughout the rest of the paper, we assume $\povmset$ is the set of all single-qubit measurements. On top of the POVM constraints, we also consider different settings of measurement selection. 
\paragraph{Adaptive.} Adaptive measurements are the most general setting, where the measurement choice can depend on previously observed outcomes. In other words, $\POVM^{(i)} = g_i(x^{i-1})$, where $g_i$ is a randomized function that maps the $i$ outcomes to a measurement in $\povmset$. Adaptivity allows us to adjust our scheme to account for variability in the measurement outcomes. The drawback of such an approach is that the measurements must run sequentially, which can cause additional computational overhead.
\paragraph{Randomized, non-adaptive.} A weaker setting to adaptivity is when the set of measurements can only depend on a shared random seed. Any non-adaptive, randomized protocol can be written as $\POVM^n = f(U)$ for a fixed random seed $U$. Although not as strong as adaptivity, the use of randomness allows the scheme to be robust to adversarial selections of the input state. On the other hand, sampling from this shared random seed might not be computationally feasible depending on the seed chosen.
\paragraph{Fixed/Determinstic.}  A fixed/deterministic measurement scheme forces the observer to fix $\POVM^{n}$ before obtaining the copies of $\sigma$. Although simple,  deterministic schemes suffer from being not robust to adversarially chosen states. \footnotemark[1]
\footnotetext[1]{Every randomized POVM can be represented as a point in the convex hull of the restricted measurement set. Thus, privately random protocols are equivalent to fixed/deterministic measurements.}
\subsection{Results and Contributions}
We will highlight the main results in our paper.  $\Tilde{O}$ and $\Tilde{\Omega}$ will denote the omission of $poly(\nqubits)$ factors for big-O and big-Omega notations, respectively. In~\cref{sec:ub-rand}, we present a non-adaptive, randomized Pauli mixedness testing algorithm.
\begin{theorem}[Randomized Pauli Mixedness Tester]\label{thm:pauli-ub} Mixedness testing can be performed with 
\begin{align*}
n = \tilde{\mathcal{O}}\left(\frac{\sqrt{10}^\nqubits}{\eps^2} \log \frac{1}{\delta} \right) 
\end{align*}
copies using non-adaptive, randomized Pauli basis measurements. 
\end{theorem}
The algorithm uses ideas from Fourier analysis of boolean functions to construct a collision-based tester to estimate the influence of fixed-order Fourier coefficients of the measurement outcome distributions. Furthermore, we show that the algorithm is tight up to $poly(N)$ factors among all single-qubit schemes with copy-wise adaptivity.
\begin{theorem}[Adaptive Single-Qubit Mixedness Testing Lower Bound]\label{thm:sq-lb} The total number of copies needed for adaptive single-qubit mixedness testing is
\begin{align*}
n=\tildeOmega{\frac{\sqrt{10}^\nqubits}{\eps^2}}.
\end{align*}
\end{theorem}
The lower bound comes from an adaptive extension of the measure-dependent dependent framework for single-copy certification~\cite{liu2024restricted}. We denote $\Luders_{\POVM}$ as the channel induced by mapping of a density matrix to the projected post-measurement state after performing measurement $\POVM$, which we refer to as the \emph{Measurement Information Channel} (MIC). In~\cref{sec:gen-bound}, we show that mixedness testing with adaptivity relies on the spectral quantities of the MIC.
\begin{theorem}[Lower Bound Framework for Adaptive Mixedness Testing]\label{thm:gen-lb} To perform mixedness testing with adaptively chosen measurements from $\povmset$, we need at least
\begin{align*}
    n = \bigOmega{\frac{4^\nqubits}{\eps^2  \sqrt{\opnorm{\Luders_\povmset} \|\Luders_\povmset\|_1}} }
\end{align*}    
total copies, where $\|\Luders_\povmset\|_{*} = \sup_{\POVM \in \povmset} \|\Luders_{\POVM}\|_*$.
\end{theorem}
\noindent By H\"older's inequality, we see that $\|\Luders_\povmset\|_{HS} \leq \sqrt{\opnorm{\Luders_\povmset} \|\Luders_\povmset\|_1}$, which results in smaller lower bound guarantee than the non-adaptive setting~\cite{liu2024restricted}. An interesting corollary of this framework is a lower bound for mixedness testing under adaptive $k$-outcome measurements, i.e. POVMs that are restricted to $k$ outcomes. Using~\cite[Lemma 2.3]{liu2024restricted} and the fact that $\opnorm{\Luders_\povmset} \leq 1$, we have the following lower bound for adaptive $k$-outcome protocols.
\begin{corollary}[Adaptive $k$-outcome Mixedness Testing Lower Bound] The total number of copies needed for Adaptive $k$-outcome mixedness testing is $n = \bigOmega{\frac{4^\nqubits}{\eps^2 \sqrt{\min\{k, 2^\nqubits\}}}}$.
\end{corollary}
Combined with~\cite[Theorem 1.1]{liu2024restricted}, we see that adaptivity does not play a significant role in mixedness testing with both single-qubit and $k$-outcome measurements. However, we show in~\cref{sec:lb-fixed} that randomness is crucial for copy-efficient single-qubit mixedness testing.
\begin{theorem}[Mixedness Testing under Fixed Pauli Measurements]\label{thm:det-pauli} The copy of complexity of deterministic Pauli mixedness testing is 
\begin{align*}
n = \bigOmega{\frac{\sqrt{27}^\nqubits}{\eps^2}},\; n = \bigO{\frac{\sqrt{30}^\nqubits}{\eps^2} \log \frac{1}{\delta}}.
\end{align*}
\end{theorem}
\noindent The theorem demonstrates the significant role of randomness for Pauli mixedness testing.
In~\cref{sec:ub-fixed}, we provide a simple protocol with fixed Pauli basis measurements and~\cref{sec:lb-fixed} establishes a lower bound with an existing fixed measurement lower bound framework~\cite{Liu2024role}.

\noindent For convenience, we have highlighted all of our results in~\cref{tab:results}. 

\begin{table}
    \def\arraystretch{2}
        \centering
        \begin{tabular}{|c | c | c | c | c |}
        \hline
       Scheme  & Pauli & Restricted & $k$-outcome & Single copy \\ \hline
       {Adaptive} & \cellcolor{pastelgreen} {$\tildeTheta{\frac{\sqrt{10}^\nqubits}{\eps^2}}$} & \cellcolor{pastelgreen} $\bigOmega{\frac{4^\nqubits}{\eps^2\sqrt{\opnorm{\Luders_\povmset}\|\Luders_\povmset\|_1}}}$ & {\cellcolor{pastelgreen}$\bigTheta{\frac{4^\nqubits}{\eps^2\sqrt{\min\{k, 2^\nqubits\}}}}$} & {$\bigTheta{\frac{\sqrt{8}^\nqubits}{\eps^2}}$}\footnotemark[2]~\footnotemark[3] \\ \hline
             {Randomized} & \cellcolor{pastelgreen}{$\tildeTheta{\frac{\sqrt{10}^\nqubits}{\eps^2}}$} & $\bigOmega{\frac{4^\nqubits}{\eps^2 \hsnorm{\Luders_{\povmset}}}}$\footnotemark[4] & {$\bigTheta{\frac{4^\nqubits}{\eps^2 \sqrt{\min\{k, 2^\nqubits\}}}}$}\footnotemark[5] & {$\bigTheta{\frac{\sqrt{8}^\nqubits}{\eps^2}}$}\footnotemark[2] \\ \hline
             {Fixed} & {\cellcolor{pastelgreen}$\bigOmega{\frac{\sqrt{27}^\nqubits}{\eps^2}}$, $\bigO{\frac{\sqrt{30}^\nqubits}{\eps^2}}$} & $\bigOmega{\frac{8^\nqubits}{\eps^2 \tracenorm{\Luders_{\povmset}}}}$\footnotemark[4] &{$\bigTheta{\frac{8^\nqubits}{\min{k, 2^\nqubits} \eps^2 }}$}\footnotemark[5] &{$ \bigTheta{\frac{4^\nqubits}{\eps^2}}$}\footnotemark[4]  \\\hline
        \end{tabular}
        \caption{Copy complexity of Mixedness Testing under different measurement models. Contributions from this paper are highlighted in green.}
        \label{tab:results}
\end{table}
\footnotetext[2]{\cite{BubeckC020}}
\footnotetext[3]{\cite{Chen0HL22}}
\footnotetext[4]{\cite{Liu2024role}}
\footnotetext[5]{\cite{liu2024restricted}}

% \noindent We also established a measurement-dependent lower bound for adaptive single-copy mixedness testing.
% \begin{theorem}[Lower Bound for Adaptive, Restricted Single-Copy Mixedness Testing ]\label{thm:gen-lb}Consider $\ell$ orthonormal operators $V_1,V_2,\ldots,V_\ell \in \mathbb{C}^{d \times d}$ for $\ell \geq d^{3/2}$. Then, an adaptive single-copy mixedness testing protocol with the restricted measurement set $\povmset$ must satisfy
%     \begin{align*}
%         n = \bigOmega{\frac{1}{\eps^2} \cdot \left(\frac{\ell^2}{ \sup_{\POVM \in \povmset} \sum_{m=1}^\ell \dotprodH{\Luders_\POVM(V_m)}{V_m}}\right)^{1/2}}
%     \end{align*}
% when $N \geq K$, where $K$ is an absolute constant dependent on $\povmset$.
% \end{theorem}

\subsection{Related Work}
\label{sec:related}

\paragraph{Quantum state tomography.}
Quantum state tomography aims to reconstruct an unknown quantum state from repeated preparations. For low-rank states, compressed-sensing techniques showed that substantially fewer measurement settings suffice under suitable measurement ensembles~\cite{GrossLFBE10,Flammia_2012,Liu2011universal,KRT14}. Entangled measurements can further reduce the sample complexity: nearly sample-optimal bounds for rank-$r$ states under trace distance were established in~\cite{HaahHJWY17,ODonnellW16,ODonnellW17}, and recent lower bounds complete the characterization, showing that $\bigTheta{rd/\varepsilon^2}$ copies are necessary and sufficient with arbitrary collective measurements~\cite{scharnhorst2025}. Related work considers alternative loss functions and estimation objectives, including fidelity and quantum $\chi^2$ divergence~\cite{Yuen_2023,flamian2023tomography,wang2024sampleoptimal}.

Entangled measurements may require more complicated procedures to prepare measurements that span across multiple copies, motivating restricted models in which copies are measured separately, possibly adaptively. An $\bigOmega{dr^2}$-type lower bound for single-copy tomography of rank-$r$ states was established in~\cite{HaahHJWY17}, with subsequent work developing lower bounds for adaptive and other restricted single-copy measurements~\cite{10.1145/3717450,chen2023does}. For general states, the tradeoff between measurement entanglement and copy complexity was further characterized in~\cite{Chen0L24memory}, interpolating between single-copy and fully collective tomography.

For $\nqubits$-qubit systems, Pauli measurements, a class of single-qubit measurements, are particularly important because of their experimental simplicity. Their hierarchical structure, where each Pauli measurement simultaneously provides information about its lower-weight marginals, was exploited to obtain an $\bigO{10^\nqubits/\varepsilon^2}$ upper bound and later to establish a separation between Pauli and unrestricted single-copy tomography~\cite{Yu2020Pauli,ADLY2025Paulinot}. More recently, a nearly matching $\bigOmega{10^\nqubits/(\sqrt{\nqubits}\varepsilon^2)}$ lower bound was proved even for arbitrary adaptive single-qubit measurements~\cite{acharya2025single}. In the pure state setting, it was recently shown that, rather surprisingly, Pauli measurements suffice to nearly meet the optimal rate of entangled and single-copy measurements $\tildeTheta{2^\nqubits/\eps^2}$~\cite{grewal2026paulpure}.

\paragraph{Quantum state certification.}
Quantum state certification, or quantum identity testing, asks whether an unknown state $\rho$ equals a known reference state $\sigma$ or satisfies $\|\rho-\sigma\|_1 \geq \eps$, and is a basic problem in quantum property testing~\cite{MdW13}. With unrestricted entangled measurements, worst-case certification requires $\Theta(2^\nqubits/\varepsilon^2)$ copies~\cite{BadescuO019, OW15}, substantially fewer than full tomography. Related work studies spectral-property testing and other quantum testing problems~\cite{OW15,BadescuO21,Yu21sample}. Measurement restrictions again increase the complexity: entangled measurements are necessary for optimal quantum property testing~\cite{BubeckC020}, while randomized single-copy certification has tight worst-case complexity $\Theta(\sqrt{8}^\nqubits/\varepsilon^2)$~\cite{Chen0HL22}. Instance-dependent bounds depending on the spectrum of the reference state were developed in~\cite{ChenLO22instance}. Randomness itself is also a resource: deterministic single-copy schemes require $\Theta(4^\nqubits/\varepsilon^2)$ copies in the worst case~\cite{Liu2024role}. Recent work further studies instance-optimal certification with entangled measurements~\cite{odonnell2025instanceocert}, certification with limited entanglement across copies~\cite{wadhwa2026optimal}, and certification under non-adaptive restricted measurements~\cite{liu2024restricted}.

\paragraph{Estimation of state properties.}
Another related line of work estimates selected properties without reconstructing the density matrix directly. Such approaches include protocols based on sampled Pauli observables and experimentally motivated certification procedures~\cite{flamia2011direct,PhysRevLett.107.210404,AGKE15}. Shadow tomography and classical-shadow methods instead predict many observables from comparatively few measurements~\cite{Aaronson20,huang2020predicting,Elben_2022}. These methods can be substantially more sample-efficient for targeted observables, but do not in general yield a global trace-distance approximation of the unknown state.

\section{Technical Overview}
In this section, we focus only on the high-level ideas of our Pauli mixedness testing algorithm, which is our key contribution. The lower bound techniques are presented in~\cref{sec:gen-bound}.

\subsection{Detecting Influence from Overlapping Measurements}
The core of our Pauli mixedness tester is to determine if the norm of the Pauli coefficients are large enough. By the typical relationship between $\ell_1$ and $\ell_2$ distances, it is clear that if $\|\sigma - \qmm\|_1 \geq \eps$ when $\sigma = \sum_{Q\in \mathcal{Q}} \frac{\alpha_Q}{d} Q$ for the set of Pauli observables $\mathcal{Q}=\{X,Y,Z,\eye\}^{\nqubits}\setminus \{\eye_d\}$, then
\begin{align*}
   \sum_{Q \in \mathcal{Q}}  \alpha_Q^2 \geq \eps^2,
\end{align*}
This norm will be $0$ when $\sigma = \qmm$. Thus, it suffices to detect if the norm of Pauli coefficients is $\geq \eps^2$ or $0$. In the two-outcome Pauli observable setting~\cite{Yu2023almost}, we can only obtain information about one Pauli coefficient at a time, leading to a copy complexity of $\tildeO{4^\nqubits/\eps^2}$. Pauli basis measurements are much stronger, enabling us to learn $2^\nqubits$ Pauli coefficients under a single measurement basis. The main question we ask is:
\begin{center}
    Can we detect the norm of these coefficients with $o(4^\nqubits/\eps^2)$ copies?
\end{center}
To answer this question, we must consider the overlapping structure of Pauli measurements. For a particular Pauli basis $P \in \{X,Y,Z\}^{\otimes \nqubits}$, we learn about the coeffcients $Q \triangleleft P$, which denotes that the non-identity tensor entries of $Q$ match $P$.  Consider an example with a $2$-qubit state. If we want to learn $\alpha_{X\eye}$, then we can choose measurement bases $XX,XY,XZ$. On the other hand, learning $\alpha_{XX}$ requires the use of measurement basis $XX$.  As a result, coefficients under the Pauli measurement structure can not be treated the same.

One basic test would be to estimate the Pauli coefficients individually and compute the norm of the coefficients based on the measurement structure. We can do so by returning the estimate
\begin{align*}
    \sum_{P} \sum_{Q \triangle P} \frac{C_Q \cdot C'_Q}{3^{\nqubits-w(Q)}},
\end{align*}
where $C_Q, C'_Q$ are independent and identical unbiased estimates of $\alpha_Q$ under multiple observations under the measurement outcome distribution. The \emph{Pauli weight} of Q, denoted as $w(Q)$, is the number of non-identity entries in $Q$. In~\cref{sec:ub-fixed}, we analyze correlations between the $C_Q$'s and show that $n=\bigO{\sqrt{30}^\nqubits/\eps^2}$ copies is sufficient to test mixedness.

The underlying issue with the estimator is that we perform a deterministic scheme on the measurement outcomes while knowing that randomness plays a significant role in state certification tasks~\cite{Liu2024role}. In fact, we prove in~\cref{sec:lb-fixed} that $\bigOmega{\sqrt{27}^\nqubits/\eps^2} = \bigOmega{4^\nqubits/\eps^2}$ copies are necessary under deterministic Pauli protocols. 

 We propose the use of basis-dependent statistics to help us take advantage of randomness under the structure of Pauli measurements.

\begin{definition}[Pauli influence] Given a Pauli basis $P \in \{X,Y,Z\}^{\otimes \nqubits}$ and a state $\sigma$ with Pauli decomposition $\{\alpha_Q\}_{Q \in \mathcal{Q}}$, the \emph{Pauli Influence} of $P$ is defined as
\begin{align*}
    L_P \eqdef \sum_{\emptyset \subset S \subseteq [N]} 3^{|S|} \alpha_{P}(S)^2 = \sum_{Q \triangleleft P} 3^{w(Q)} \alpha_Q^2
\end{align*}
where $\alpha_P(S)$ is the Pauli coefficient of the observable that matches entries of $P$ for $i \in S$ and is $\eye$ for $i \notin S$ and $w(Q)$ is the number of non-identity tensor factors of $Q$.
\end{definition}
 We scale each Pauli coefficient monotonically by the Pauli weight. The idea is that $3^{w(Q)}$ gives us a priority on which Pauli coefficients should be tested with higher precision. Revisiting our prior example, the low weight $\alpha_{\eye X}$ will have priority of $3$ while the high weight $\alpha_{XX}$ will have a priority of $9$ . An important fact of Pauli influence is that it captures the total norm of the coefficients when $P$ is sampled uniformly. As a result, we have the following property,
% Given this prioritization of the coefficients, we need to ensure that sampling $\tildeO{J}$ measurements\ymargin{Unclear} uniformly results in a measurement of total influence of $\bigOmega{J \|\alpha_Q\|_2^2}$, as we need to control the total amount of ``work'' used to detect the non-zero groups. This approach is widely known as \emph{Levin's Work Investment Strategy}~\cite{levin1985oneway}. It spans various theoretical computer science domains, including property testing~\cite{goldreich2014input}, graph theory~\cite{goldreich2002testing}, quantum state certification~\cite{Yu2023almost}, and sublinear algorithms~\cite{berman2014lp}. We will prove the following Levin's work strategy result for detecting Pauli influence.
\begin{lemma}[Detectability of Pauli Group Influence]\label{lem:pauli_work_invest}
    Consider a mixed state $\rho$ with Pauli decomposition $\{\alpha_Q\}_{Q \in \mathcal{Q}}$. Let $L_P = \sum_{Q \triangleleft P} 3^{w(Q)}\alpha_Q^2$ and $L = \sum_{Q \in \mathcal{Q}} \alpha_Q^2$, then there exists $i \in [N]$ such that
    \begin{align*}
       \Pr[L_P \geq 3^{i-2} L] \geq \frac{3^{-i}}{N},
    \end{align*}
    where $P \sim \unif[\{X,Y,Z\}^{\otimes N}]$.
\end{lemma}
At a high-level, there exists a level $i$ where we can perform $\tildeO{3^i}$ amount of work to detect influence with threshold $\Omega(3^i \|\alpha_Q\|_2^2)$. This approach is widely known as \emph{Levin's Work Investment Strategy}~\cite{levin1985oneway}. It spans various theoretical computer science domains, including property testing~\cite{goldreich2014input}, graph theory~\cite{goldreich2002testing}, quantum state certification~\cite{Yu2023almost}, and sublinear algorithms~\cite{berman2014lp}. This property enables us to randomly search by uniformly sampling Pauli bases and allocate multiple copies for each basis to test if its influence supersedes a pre-defined threshold.

Say we find a $P$ with influence $\bigOmega{3^i \|\alpha_Q\|_2^2}$. How would we test for $L_P$? It is natural to try to directly compute the statistic $L_P$ from multiple observations under the outcome distribution of basis $P$. However, the issue is that the statistics for computing $\alpha_{X \eye}$ and $\alpha_{XX}$ can be highly correlated, resulting in a potential blowup in variance. We overcome this hurdle by splitting the statistic by weight. In particular, there exists a weight $w \in [N]$ that has a large norm:
\begin{align*}
\sum_{|S|=w} \alpha_P(S)^2 \geq \Delta_w  = \frac{3^i \|\alpha_Q\|_2^2}{3^w N}.
\end{align*}
In fact, we can see that the norm of these $\binom{\nqubits}{w}$ coefficients dominates the total norm of the $2^\nqubits$  coefficients by factor $\tildeOmega{3^{i-w}}$. Taking advantage of the following fact, we can actually perform the weight test for each $w \in [\nqubits]$ with $\bigO{\sqrt{\binom{N}{w}}/\Delta_w}$ allocated copies for the fixed Pauli basis $P$. This result is reminiscent of Gaussian mean testing, where classically testing whether $X \sim \mathcal{N}(0, \eye_d)$ or $X \sim \mathcal{N}(\mu, \eye_d)$ can be done with $\bigTheta{\sqrt{d}/\|\mu\|_2^2}$ samples, despite the statistics of the Pauli coefficients being highly correlated.

In order to attain these results, we make an important observation regarding the relationship between the Pauli coefficients and the outcome distribution of a Pauli basis measurement.
\begin{observation}[Fourier Expansion of Pauli Outcome Distributions] \label{obs-pauli-fourier}
The outcome distribution of measurement $P$ on state $\sigma$ with Pauli decomposition $\{\alpha_Q\}_{Q \in \mathcal{Q}}$ is
\begin{align*}
r_P(x) = \frac{1}{2^N} \sum_{S \subseteq [\nqubits]} \Tr[P^S \rho] \cdot \parity{x} = \frac{1}{2^N} + \frac{1}{2^N} \sum_{\emptyset \subset S \subseteq [\nqubits]} \alpha_{P}(S) \cdot \parity{x}
\end{align*}
\end{observation}
The Fourier expansion of the outcome distribution under $P$ is precisely the Pauli coefficients under $Q \triangleleft P$. In other words, $\alpha_{P}(S)$ corresponds to correlation under the set $S$ for outcome distributions of the measurements $P^S \triangleleft P$. We  leverage ideas in the analysis of Boolean functions to test weight $w$ coefficients of an unknown distribution.

\subsection{Correlation-Concentrated Uniformity Testing via Polynomial Projection}
Recall, the goal is to try to test the norm of weight $w$ Fourier coefficients of a distribution supported on the Boolean hypercube. This aforementioned problem has strong ties to uniformity testing, as Parseval's theorem states that $d \|p-U_d\|_2^2 = \|\alpha\|_2^2$, where $\alpha(S)$ are the coefficients of the Fourier expansion of $p$. It is well known that the optimal sample complexity of uniformity testing is $\bigTheta{\sqrt{d}/\eps^2}$ \cite{DiakonikolasGPP19}, which is achieved by estimating the collision probability $\|p\|_2^2$. If we simply used this subroutine to test the norms of each weight in $L_P$, then we can only prove a bound that is worse than $4^\nqubits/\eps^2$. This suggests that we need to leverage concentration on weight $w$ coefficients to perform uniformity testing with $o(\sqrt{d}/\eps^2)$ samples. 

The notion of concentration that we consider is defined by $(\Delta, k, \beta)$-wise correlatedness, which holds if a distribution $p$ follows $\sum_{|S|=k} \alpha(S)^2 \geq \beta \|\alpha\|_2^2$ and the norm of the weight $k$ coefficients clear some threshold $\beta \Delta$. It turns out that we can test uniformity with much fewer number of samples under these specific classes of distributions!
 
\begin{theorem}[Correlation-Concentrated Uniformity Testing]\label{thm:sparse-test} Let $\mathcal{P}^{(\Delta, k, \beta)}$ be the collection of distributions that are $(\Delta,k, \beta)$-wise correlated. Then, \cref{alg:sparse} uses $\bigO{\frac{\sqrt{\binom{N}{k}}}{\beta \Delta} \ln \frac{1}{\delta}}$ samples to test whether $p=u$ or $p \in \mathcal{P}^{(\Delta, k, \beta)}$ with probability at least $1-\delta$.
\end{theorem}
The main idea behind the tester is to refine our notion of collision testing on the hypercube. Consider the original collision-based tester using sets of disjoint samples $X_1, \ldots, X_n$ and $Y_1, \ldots, Y_n$,
\begin{align*}
    W = \frac{d}{n^2} \sum_{i=1}^n \sum_{j=1}^n \indic{X_i = Y_j},
\end{align*}
which is unbiased estimator of 1+$\|\alpha\|_2^2$. Now, we will now closely look at the indicator function $\indic{x = y}$. We can rewrite it as the following,
\begin{align*}
    d \indic{x = y} = \sum_{S \subseteq [\nqubits]}\parity{x \oplus y},
\end{align*}
as the sum over all parities will only be non-zero if $x \oplus y = 1^{\nqubits} \Leftrightarrow x=y$, where $\oplus$ is a elementwise product. To see this, take a vector $z$ that has $z_i = -1$, then we can split the sum into two parts:
\begin{align*}
\sum_{S\subseteq[\nqubits]} \parity{z} = \sum_{i \in S} \parity{z} + \sum_{i \notin S} \parity{z}. 
\end{align*}
It is clear that $\sum_{S\subseteq[\nqubits]} \parity{z} = 0$ , because $\sum_{i \notin S} \parity{z} =  - \sum_{i \in S} \parity{z}$ whenever $z$ contains a $-1$ coordinate. The indicator function can be seen as the sum of all set-wise correlations of the collided vector $x \oplus y$, and it enables us to learn about $\|\alpha\|_2^2$. What if we restricted this function to evaluate set-wise correlations of weight $k$? Then, the function is described by the Krawtchouk Polynomial with symmetry in Hamming distance: $d_H(x,y) = \sum_{i=1}^\nqubits \indic{x_i \neq y_i}$.
\begin{align*}
    \sum_{|S|=k}\parity{x \oplus y} = \krawt{k}{d_H(x, y)}, \;\krawt{k}{t} = \sum_{j=0}^k (-1)^j \binom{t}{j} \binom{N-t}{k-j}.
\end{align*}
Thus, the indicator can be interpreted as a polynomial of degree $\nqubits$.
\begin{align*}
     d \indic{x = y} = \sum_{k=1}^{\nqubits} \krawt{k}{d_H(x, y)}.
\end{align*}
Since the degree $\nqubits$ polynomial of degree $\nqubits$ learns about $\|\alpha\|_2^2$, it would intuitively make sense that the polynomial of order $k$ will learn about $\|\alpha^{=k}\|_2^2$. In other words, we can try to estimate $\|\alpha^{=k}\|_2^2$ with
\begin{align*}
     W_k = \frac{1}{n^2} \sum_{i=0}^n \sum_{j=1}^n \krawt{k}{d_H(X_i, Y_j)}.
\end{align*}
Using the properties on convolution and inner product, we can show that $W_k$ is indeed an unbiased estimator of $\|\alpha^{=k}\|_2^2$. Furthermore, we can exploit the sparseness and symmetry of the Krawtchouk polynomial to show that $W_k$ only suffers from a standard deviation scaling of $\sqrt{\binom{N}{k}}$ as opposed to the $\sqrt{d}$ scaling of $W$. On top of being sample-efficient, we also can show that the tester is computationally effieicent,
\begin{corollary}[Efficient Correlation-Concentrated Uniformity Testing] \label{cor:eff-sparse} $W_k$ from~\cref{alg:sparse} admits an implementation running in $\bigO{n^2 \cdot k\nqubits\log\nqubits}$ time under the bit complexity model. Thus, a correlation-concentrated uniformity tester can be implemented in $poly(n, N, k)$ time.
\end{corollary}
\vspace{-0.5cm}
\noindent \paragraph{Comparison to other approaches.} A similar problem to correlation-concentrated testing is testing $k$-wise uniformity, where the goal is to determine whether a distribution behaves like a uniform distribution across $k$ coordinates or $\epsilon$-far from every such distribution. A norm estimator has been proposed to solve this problem~\cite[Theorem 9]{odonnell2018closenesskwiseuniformity}, and it yields optimal results when assuming $k$ is constant. The estimator estimates the norm of the coefficients from $1,2,\ldots k$ by directly computing the parity functions of size $\leq k$,  which led to a standard deviation scaling of $\bigO{2^k N^{k/2}}$. However, it is possible refine the standard deviation scaling to $\tildeO{\sqrt{\binom{N}{k}}}$ by analyzing the binomial coefficients in the variance more precisely, which leads to comparable statistical performance to our polynomial estimator.

We highlight the computational efficiency of our approach, requiring only $\bigO{kN \log N}$ time to evaluate the Krawtchouk polynomial. On the other hand, the parity-based estimator requires $\Omega(\binom{N}{k})$ time to evaluate the sum of cross products for every parity, resulting in potentially exponential extra overhead.
\subsection{Why \texorpdfstring{$\sqrt{10}$}{sqrt(10)}?}
From the above sections, it still may not be clear why the $\sqrt{10}$ factor appears for mixedness testing. We show that this number naturally arises from the combinatorial factors introduced by correlation-concentrated testing and Pauli influence detection. Recall that we are guaranteed to obtain at least one ``good'' measurement with influence $L_P = \Omega(3^i \eps^2)$ when sampling $\tildeO{3^i}$ measurement bases for some $i \in [\nqubits]$. For this particular measurement basis, we can guarantee that one of the Pauli weights dominates the total influence by factor $\tildeOmega{3^{i-w}}$.

For each sampled Pauli basis, we use our correlation-concentrated testing subroutine to test each isolated weight of the influence with $\tildeO{3^{w-i} \sqrt{\binom{\nqubits}{w}}/\eps^2}$ samples/copies for every $w \in [\nqubits]$, including additional log-factors needed for every test in level $i$ to pass with high probability. Now, we can tally up the total number of copies of $\rho$ used for our mixedness tester. Considering that we must test every $i \in [\nqubits]$ for the work investment strategy,
\begin{align*}
n = \sum_{i=1}^\nqubits \sum_{w=1}^\nqubits \tildeO{\frac{3^w \sqrt{\binom{\nqubits}{w}} }{3^i \eps^2} \cdot 3^i} = \tildeO{\frac{\max_{w} c(w)}{\eps^2}}.
\end{align*}
Now, it suffices to optimize over the combinatorial quantity: $c(w) = 3^w \sqrt{\binom{\nqubits}{w}} = \sqrt{9^w \binom{\nqubits}{w}}$. Using the fact that $\sum_{w=1}^\nqubits 9^w \binom{\nqubits}{w} = 10^\nqubits$, we can show that $c(w) \leq \sqrt{10}^\nqubits$. This is nearly tight since $c(9\nqubits/10) = \tildeOmega{\sqrt{10}^\nqubits}$ by Stirling's approximation. Thus, we have reached the conclusion that
\begin{align*}
    n = \tildeO{\frac{\sqrt{10}^\nqubits}{\eps^2}}.
\end{align*}
Interestingly enough, the weight $9\nqubits/10$ holds a significant meaning in the lower bound as well. The hard-instance construction concentrates perturbations around the maximally mixed state with Pauli observables of weight at least $9\nqubits/10$. This gives us the crucial intuition that, like single-qubit tomography~\cite{acharya2025single}, the copy-complexity of single-qubit mixedness testing is dictated by the observables of weight $9\nqubits/10$!

\section{Preliminaries}
\subsection{Basic quantum mechanics}
An isolated physical system is associated with a
Hilbert space is called the { state space}. A { pure state} of a
quantum system is a normalized vector in its state space, denoted by the Dirac notation $\ket{\varphi}$. A
{\it mixed state} is represented by a density operator on the inner product space with $\dotprodH{\rho}{\sigma} = \tr(\rho \sigma)$. Here, a density operator $\rho$ on $d$-dimensional Hilbert space $\C^d$ is a
positive semi-definite linear operator such that $\Tr(\rho)=1$.
We let
\begin{align*}
\mathcal{D}({\C^{d}})=\{\rho: \rho~\mathrm{is}~d\mathrm{-dimensional~density~operator~of~}\C^d\}
\end{align*}
denote the set of quantum states.

We use $\eye$, $X, Y$, and $Z$ to denote the Pauli matrices,
\begin{align*}
\eye=\begin{bmatrix}1 &0\\0&1\end{bmatrix}, X=\begin{bmatrix}0 &1\\1&0\end{bmatrix}, Z=\begin{bmatrix}1 &0\\0&-1\end{bmatrix}, Y=\begin{bmatrix}0 &i\\-i&0\end{bmatrix}.
\end{align*}

A important fact is that $\{I,X,Y,Z\}$ is a orthogonal basis for mixed states. Any mixed state $\rho$ can be written as
\begin{align*}
    \rho = \frac{\eye}{2} + \frac{\alpha_X}{2} X + \frac{\alpha_Y}{2} Y + \frac{\alpha_Z}{2} Z.
\end{align*}
We denote $\{\alpha_Q\}_{Q \in \{X,Y,Z\}}$ to be the \emph{Pauli decomposition} of $\rho$. Note, we fix $\alpha_\eye = \frac{1}{2}$ to ensure $\tr(\rho) = 1$ and $\sum_{P\in \{X,Y,Z\}} \alpha_P^2 \leq 1$ as a necessary condition for positive semi-definiteness.

\subsection{The tensor product of Hilbert space}

The state space of a composed quantum system is the tensor product of the state spaces of its component systems.
Let $\C^{d_k}$ be a Hilbert space. One can define a Hilbert space $\bigotimes_{k=1}^{n}\C^{d_k}$ as the tensor product of Hilbert spaces $\C^{d_k}$.
The following notation denotes the quantum state on the multipartite system $\bigotimes_{k=1}^{n}\C^{d_k}$,
\begin{align*}
\mathcal{D}(\otimes_{i=1}^{\nqubits}\C^{d_i})=\mathcal{D}(\C^{\Pi_{i=1}^{\nqubits} d_i}).
\end{align*}
An interesting consequence of tensor product is that a basis of $\mathcal{D}(\otimes_{i=1}^{\nqubits}\C^{d_i})$ can be written as the tensor product of the basis for each $\mathcal{D}(\C^{d_i})$. Namely, if we used the Pauli basis for a single-qubit, then $\mathcal{Q} = \{\eye,X,Y,Z\}^{\otimes \nqubits} \setminus \{\eye_d\}$ serves as a basis for $\mathcal{D}(\C^d)$ when $d = 2^N$. Any $\nqubits$-qubit state $\rho \in \mathcal{D}(\mathbb{C}^d)$ can be written as
\begin{align*}
    \rho = \sum_{Q \in \mathcal{Q}} \frac{\alpha_Q}{d} Q.
\end{align*}
Here, we have the Pauli decomposition $\{\alpha_Q\}_{Q \in \mathcal{Q}}$ and similar constraints to the single-qubit decomposition with $\alpha_{I_d} = \frac{1}{d}$ and $\sum_{Q \in \mathcal{Q}} \alpha_Q^2 \leq d - 1$. We will call $\qmm = \frac{\eye_{d}}{d}$ the maximally mixed state of $\mathcal{D}(\C^{d})$.

\subsection{Quantum measurement}

\paragraph{Positive operator-valued measures. } A positive operator-valued measure (POVM) is a measurement with non-negative self-adjoint operators in a Hilbert space $\C^{d}$, which is
described by a collection of matrices $\{M_i\}$ with $M_i\geq 0$ and
\begin{align*}\sum_{i}M_i=I_{d}.\end{align*}
If the state of a quantum system was $\rho$
 before measurement $\{M_i\}$ was performed,
the probability of $i$ occurring is
\begin{align*}p(i)=\dotprodH{M_i}{\rho} = \tr(M_i\rho).\end{align*}
A Hermitian $O=\sum_j \lambda_j \op{\psi_j}{\psi_j}$ (that is $O=O^{\dag}$) on $\mathcal{H}$ with orthonormal basis $\ket{\psi_j}$ and $\lambda_j\in\mathbb{R}$ always corresponds to the following measurement protocol.
Suppose the state of a quantum system is $\rho$ when measured in basis $\op{\psi_j}{\psi_j}$. If the outcome is $j$, the observed result is $\lambda_j$, and the output corresponds directly to a random variable $X$ such that
\begin{align*}
p(X=\lambda_j)&=\tr(\rho\op{\psi_j}{\psi_j}).
\end{align*}
We will denote the subset of POVMs that we are allowed to perform to be $\povmset$.

\paragraph{Measurement information channel.} The \emph{Measurement Information Channel} (MIC) describes the projected state after performing a POVM, and it serves as a vital object for quantifying the ``distinguishing power'' of a measurement~\cite{liu2024restricted}.
\begin{definition} \label{def:mic}
Let $\POVM=\{M_x\}_{x \in \mathcal{X}}$ be a POVM. Then the Measurement Information Channel of $\POVM$ is defined as
\begin{align} 
    \Luders_\POVM(\rho) = \sum_{x} \frac{\dotprodH{M_x}{\rho} M_x}{\Tr[M_x]} = \sum_{x} \frac{\Tr(M_x \rho) M_x}{\Tr(M_x) }.
\end{align}
\end{definition}
\noindent The channel is a positive and trace-preserving map on $\mathcal{D}(\mathbb{C}^d)$.

\paragraph{Pauli Measurements.} For $U\in\{X,Y,Z\}$ with $U=\op{\psi_0}{\psi_0}-\op{\psi_1}{\psi_1}$, we use the Pauli measurement to denote the corresponding measurement,
\begin{align*}
M_+=\op{\psi_0}{\psi_0}, M_-=\op{\psi_1}{\psi_1}.
\end{align*}

For an $\nqubits$-qubit Pauli matrix $P=P_1P_2\cdots P_{\nqubits}\in \{\eye,X,Y,Z\}^{\otimes \nqubits}$ with $P\neq \eye_d$, we can always find orthogonal projections $Q_{+}$ and $Q_{-}$, i.e., $Q_{+}^2=Q_{+}$ and $Q_{-}^2=Q_{-}$, such that $P=Q_{+}-Q_{-}$ and $Q_{+}+Q_{-}= \eye^{\otimes \nqubits}$. Furthermore, the two outcome measurement
$\{Q_+,Q_-\}=\{\frac{\eye_d + P}{2},\frac{\eye_d-P}{2}\}$ can be implemented through the one-qubit Pauli measurements corresponding to $P = P_1 \otimes P_2 \otimes \ldots \otimes P_{\nqubits}$.

\medskip
\begin{observation}\label{observ}
For any
\[
P = P_1 \otimes P_2 \otimes \cdots \otimes P_{\nqubits}
\in \{X,Y,Z\}^{\otimes \nqubits},
\]
performing the measurement $P_i$ on the $i$-th qubit for each
$i\in[\nqubits]$ produces an $\nqubits$-bit outcome string $x \in \{-1,1\}^\nqubits$. The following POVM describes this Pauli basis measurement,
\begin{align*}
    M^{P}_x = \bigotimes_{i=1}^\nqubits \frac{\eye + x_i P_i}{2}.
\end{align*}

\noindent Moreover, the same outcome $x$ can also be interpreted as the
measurement outcome of any
\[
Q = Q_1 \otimes Q_2 \otimes \cdots \otimes Q_{\nqubits}
\in \{I,X,Y,Z\}^{\otimes \nqubits},
\]
provided that, for every $i\in[n]$, either $Q_i = P_i$ or
$Q_i = I$.
We say that such a Pauli operator $Q$ \emph{corresponds to} $P$, denoted as $Q\triangleleft P$. 

We will also denote $Q = P^{S}$ for $S \subseteq [\nqubits]$ to be the Pauli matrix where $Q_i = P_i$ if $i \in S$ and $Q_i = \eye$ otherwise and provide the shorthand $\alpha_P(S) = \alpha_{P^S}$. Indexing over all $S \subseteq [\nqubits]$ gives us all the Pauli operators that correspond to $P$.

\end{observation}

\subsection{Distances between states} \label{pre:dist}
\noindent For indexable set of numbers $t(i)$ for $i \in \mathcal{I}$, we refer to its $\ell_1$ norm as $\|t\|_1 = \sum_{i \in \mathcal{I}} |t(i)|$ and $\|t\|_2 = \sqrt{\sum_{i \in \mathcal{I}} |t(i)|^2}$ for its $\ell_2$ norm. Given a general operator $A$ with eigenvalues $\lambda_1, \ldots, \lambda_d$, the $\ell_1$ norm is defined as
\begin{align*}
||A||_1=\mathrm{Tr}|A| = \|\lambda\|_1,
\end{align*}
where $|A|\equiv\sqrt{A^\dag A}$ is the positive square root of $A^\dag A$. The $\ell_2$ norm is defined as
\begin{align*}
    ||A||_2=\sqrt{\mathrm{Tr}(A^2)} = \|\lambda\|_2.
\end{align*}
We denote the corresponding $\ell_1$ and $\ell_2$ distances to be $\|\rho - \sigma\|_1$ and $\|\rho - \sigma\|_2$, respectively.
The $\ell_2$ distance has a particularly nice form in terms of the Pauli coefficients due to the fact that $\mathcal{Q}$ (with normalization factor $\frac{1}{\sqrt{d}}$) form an orthogonal basis of $\mathcal{D}(\mathbb{C}^d)$. Let $\{\alpha_Q\}_{Q \in \mathcal{Q}}$ and $\{\beta_Q\}_{Q \in \mathcal{Q}}$ be the Pauli decompositions of $\rho$ and $\sigma$. Then,
\begin{align*}
    ||\rho-\sigma||_2^2 = \sum_{Q \in \mathcal{Q}} \frac{(\alpha_Q - \beta_Q)^2}{d} = \frac{1}{d} \|\alpha - \beta\|_2^2,
\end{align*}
For $\rho,\sigma\in\mathcal{D}(\C^{d})$, we have the following relation between $\ell_1$ and $\ell_2$ distances,
\begin{align*}
||\rho-\sigma||_2\leq ||\rho-\sigma||_1\leq \sqrt{d}||\rho-\sigma||_2 = \|\alpha-\beta\|_2.
\end{align*}
We also have the relation $\tracenorm{\rho - \sigma} \leq \sqrt{2r} \|\rho - \sigma\|_2 $ for $r = \max\{rank(\rho), rank(\sigma)\}$.
\subsection{Fourier analysis of Boolean Functions}
We will go over some basic facts regarding Fourier analysis of Boolean functions. For a more comprehensive review of the subject, we refer the reader to the excellent book \cite{odonnell2021analysisbooleanfunctions}.
\paragraph{Fourier Identities.} Consider a real-valued function on the Boolean hypercube $f: \{-1,1\}^\nqubits \rightarrow \mathbb{R}$. Define the parity functions on the hypercube as $\parity{x} \eqdef \prod_{s \in S} x_i$ for $S \subseteq [\nqubits]$. The parity functions serve as an orthogonal basis under the inner product space $\dotprodB{f}{g} \eqdef \expectDistrOf{\unif[\{-1,1\}^\nqubits]}{f(x)g(x)}$. Thus, $f$ can be uniquely expressed as a multi-linear polynomial:
\begin{align*}
    f(x) = \sum_{S \subseteq [N]} \alpha(S) \cdot \parity{x}.
\end{align*}
A consequence of these Fourier expansions is Parseval's Theorem. Given $g(x) = \sum_{S \subseteq [N]} \beta(S) \parity{x}$,
\begin{align*}
    \dotprod{f}{g}_{B} = \sum_{S \subseteq [\nqubits]} \alpha(S) \beta(S).
\end{align*}
Let's define the convolution of two functions $f * g (x) \eqdef \expectDistrOf{\unif[\{-1,1\}^\nqubits]}{f(y)g(x \oplus y)}$, where we denote $x \oplus y$ as element-wise multiplication. We have the following Fourier-analytic identity,
\begin{align*}
    f * g (x) = \sum_{S \subseteq [\nqubits]} \alpha(S) \beta(S) \cdot \parity{x}.
\end{align*}
\paragraph{Distributions on the Hypercube.} A distribution on the Boolean hypercube has the following Fourier expansion,
\begin{align*}
p(x) = \frac{1}{2^\nqubits} + \frac{1}{2^\nqubits} \sum_{\emptyset \subset S \subseteq [\nqubits]} \alpha(S) \cdot  \parity{x},
\end{align*}
where $\alpha(S) \in [-1,1]$ are chosen to guarantee non-negativity over $x \in \{-1,1\}^N$. The constraint $\alpha_{\emptyset} = 1$ is necessary to ensure normalization of the distribution. We denote $\|\alpha^{=k}\| \eqdef \sum_{|S| = k} \alpha(S)^2$ to be the total weight of the $k$-th order correlations. We provide a useful notion of concentration on coordinate-wise correlations.
\begin{definition}[$(\Delta,k, \beta)$-wise correlated distributions] For $\Delta, \beta \in (0,1]$, a distribution supported on $x \in \{-1,1\}^\nqubits$,
\begin{align*}
    p(x) =\frac{1}{2^\nqubits} + \frac{1}{2^\nqubits} \sum_{\emptyset \subset S \subseteq [\nqubits]} \parity{x} \alpha(S)
\end{align*}
is \emph{$(\Delta,k, \beta)$-wise correlated} if $\|\alpha^{=k}\|_2^2 \geq \beta\|\alpha\|_2^2$ and $\|\alpha^{=k}\|_2^2 \geq \beta \Delta$.
\end{definition}

\paragraph{Krawtchouk Polynomials.} Krawtchouk polynomials will serve as an important tool for testing correlations on the hypercube.
\begin{definition}[Krawtchouk Polynomials] For $t,k \in [N]$, the $k$-th order Krawtchouk Polynomial is defined as
\begin{align*}
    \krawt{k}{t} \eqdef \sum_{j=0}^k (-1)^j \binom{t}{j} \binom{N-t}{k-j}.
\end{align*}
\end{definition}
These polynomials arise in our testing subroutines as they express the Hamming weight symmetry of functions with constant $k$-th order correlations,
\begin{align} \label{eq:k-sym}
        f(x) = \sum_{|S| = k} \parity{x} = K_k^{(N)}(w_{H}(x)),
\end{align}
where $w_H(x)$ denotes the total number of $-1$'s in x. We also define the Hamming distance between two vectors: $d_H(x,y) \eqdef \sum_{i=1}^N 1^{(1-x_i \cdot y_i)} = w_H(x \oplus y)$.

\section{Algorithms for Pauli Mixedness Testing} \label{sec:ub}
In this section, we will describe our Pauli mixedness testing algorithms. The first is a naïve approach that is deterministic, which is not-optimal but useful in settings with limited shared randomness. Then, we will present a near-optimal algorithm that performs non-adaptive, randomized measurements.
\subsection{Fixed Pauli Measurements} \label{sec:ub-fixed}
We will analyze a simple protocol for testing mixedness with fixed Pauli measurements. Given $2n \cdot 3^\nqubits$ copies, we allocate $2n$ copies for each measurement basis $P \in \{X,Y,Z\}^{\otimes \nqubits}$. Then, we construct the norm of the Pauli coefficients across each fixed measurement group $P$.
\begin{align*}
    W_P &= \sum_{\emptyset \subset S \subseteq [\nqubits]} \frac{C_P(S) \cdot C'_P(S)}{3^{\nqubits - |S|}},
\end{align*}
where $C_P(S) = \frac{1}{n} \sum_{i=1}^n \parity{X_P^{(i)}}$ and $C'_P(S) = \frac{1}{n} \sum_{i=1}^n \parity{Y_P^{(i)}}$ for $X_P^{(1)}, \ldots, X_P^{(n)}$ and $Y_P^{(1)}, \ldots, Y_P^{(n)}$ sampled from the outcome distribution after measuring the copies of $\sigma$ in basis $P$. Now, we combine all the norms to create the final statistic, 
\begin{align*}
    W &= \sum_{P \in \{X,Y,Z\}^{\otimes \nqubits}} W_P.
\end{align*}
We prove that $W$ returns a multiplicative approximation of the norm of the Pauli coeffcients with high probability when $n = \bigO{\sqrt{\frac{10}{3}}^\nqubits / \eps^2}$, thus the upper bound in~\cref{thm:det-pauli}.
\begin{proof}[Proof of upper bound in~\cref{thm:det-pauli}]
We will first show that $W$ is an unbiased estimator of $\|\alpha_Q\|_2^2$. First note
\begin{align*}
    \expectDistrOf{}{C_P(S)} &= \expectDistrOf{X_P}{\parity{X_P}} = 2^{\nqubits} \dotprodB{r_P}{\chi_S},
\end{align*}
where $r_P$ is the outcome distribution of measuring $\sigma$ under the basis $P$. Using~\cref{obs-pauli-fourier} and Parseval's theorem,
\begin{align*}
    \expectDistrOf{}{W_P} &= \sum_{\emptyset \subset S \subseteq [\nqubits]} \frac{\expectDistrOf{}{C_P(S)}^2}{3^{\nqubits - |S|}} = \sum_{\emptyset \subset S \subseteq [\nqubits]} \frac{\alpha_P(S)^2}{3^{\nqubits - |S|}}
\end{align*}
Thus,
\begin{align*}
    \expectDistrOf{}{W} &= \sum_{P \in \{X,Y,Z\}^{\otimes \nqubits}}  \sum_{\emptyset \subset S \subseteq [\nqubits]} \frac{\alpha_P(S)^2}{3^{\nqubits - |S|}} \\
    &= \sum_{P \in \{X,Y,Z\}^{\otimes \nqubits}} \sum_{\underset{Q \neq \eye_d}{Q \triangleleft P}}  \frac{\alpha_Q^2}{3^{\nqubits - w(Q)}} = \sum_{Q \in \mathcal{Q}} \alpha_Q^2.
\end{align*}
It now suffices the bound the variance of $W$. We do so by first bounding the variances of each $W_P$.
\begin{align*}
    \Var(W_P) = \sum_{\emptyset \subset S \subseteq [\nqubits]} \frac{\Var(C_P(S) \cdot C'_P(S))}{9^{\nqubits - |S|}} + \sum_{\emptyset \subset S,S' \subseteq [\nqubits]} \frac{Cov(C_P(S) \cdot C'_P(S), C_P(S') \cdot C'_P(S'))}{3^{\nqubits - |S|} \cdot 3^{\nqubits - |S'|}}
\end{align*}
Using the fact that $\Var(X \cdot X') = \Var(X)^2 + 2\expectDistrOf{}{X}^2 \cdot \Var(X)$ when $X,X'$ are identical and independent,
\begin{align*}
    \Var(C_P(S) \cdot C'_P(S)) &= \Var(C_P(S))^2 + 2 \alpha_P(S)^2 \cdot \Var(C_P(S)) \\
    &\leq \frac{1}{n^2} + \frac{2 \alpha_P(S)^2}{n},
\end{align*}
where we bounded the variance of the sum of independent $[-1,1]$ by $n$. Now we will expand the covaraince,
\begin{align*}
    Cov(C_P(S) \cdot C'_P(S), C_P(S') \cdot C'_P(S')) &= \expectDistrOf{}{C_P(S) \cdot C_P(S')}^2 - \alpha_P(S)^2 \cdot \alpha_P(S')^2
\end{align*}
We will now unravel the expectation of the cross product between sets $S, S'$.
\begin{align*}
    \expectDistrOf{}{C_P(S) \cdot C_P(S')} = \frac{1}{n^2} \sum_{i, j \in [\nqubits]} \expectDistrOf{}{\parity{X_P^{(i)}} \cdot \parity{X_P^{(j)}}}
\end{align*}
When $i = j$, then $\expectDistrOf{}{\parity{X_P^{(i)}} \cdot \parity{X_P^{(j)}}} = \expectDistrOf{}{\chi_{S \Delta S'}(X_P^{(i)})} = \alpha_{P}(S \Delta S')$. Otherwise, the expectation is $\alpha_P(S) \cdot \alpha_P(S')$. This yields us the following,
\begin{align*}
\expectDistrOf{}{C_P(S) \cdot C_P(S')}^2 &= \left[\frac{n(n-1)\alpha_P(S) \cdot \alpha_P(S')}{n^2} + \frac{\alpha_P(S \Delta S')}{n}\right]^2 \\
&\leq \alpha_P(S)^2 \cdot \alpha_P(S')^2 + \frac{2 \alpha_P(S) \alpha_P(S') \alpha_P(S \Delta S')}{n}+\frac{\alpha_P(S \Delta S')^2}{n^2}, \\
&\leq \frac{11 \alpha_P(S)^2 \cdot \alpha_P(S')^2}{10} + \frac{11 \alpha_P(S \Delta S')^2}{n^2}.
\end{align*}
where we used the fact $2ab \leq a^2 + b^2$ for the last line. We now have
\begin{align*}
     Cov(C_P(S) \cdot C'_P(S), C_P(S') \cdot C'_P(S')) \leq \frac{11 \alpha_P(S \Delta S')^2}{n^2} + \frac{\alpha_P(S)^2 \cdot \alpha_P(S')^2}{10}.
\end{align*}
Revisiting our variance,
\begin{align*}
    \Var(W_P) &\leq \frac{10^\nqubits}{9^\nqubits n^2} + \frac{1}{n} \sum_{\underset{Q \neq \eye_d}{Q \triangleleft P}} \frac{2 \alpha_Q^2}{9^{N - w(Q) }} + \sum_{\emptyset \subset S,S' \subseteq [\nqubits]} \frac{1}{3^{\nqubits - |S|} \cdot 3^{\nqubits - |S'|}} \left[\frac{11 \alpha_P(S \Delta S')^2}{n^2} + \frac{\alpha_P(S)^2 \cdot \alpha_P(S')^2}{10}\right].
\end{align*}
We will bound each term across the sum over all measurement bases.
\begin{align*}
    \frac{1}{n} \sum_{P \in \{X,Y,Z\}^{\otimes \nqubits}} \sum_{\underset{Q \neq \eye_d}{Q \triangleleft P}} \frac{2 \alpha_Q^2}{9^{N - w(Q) }}  &= \frac{2}{n} \sum_{Q \in \mathcal{Q}} \frac{\alpha_Q^2}{3^{\nqubits - w(Q)}} \\
    &\leq \frac{2 \|\alpha_Q\|_2^2}{n}.
\end{align*}
For the first covariance term,
\begin{align*}
    \sum_{P \in \{X,Y,Z\}^{\otimes \nqubits}}\sum_{\emptyset \subset S,S' \subseteq [\nqubits]} \frac{\alpha_P(S)^2 \cdot \alpha_P(S')^2}{10 \cdot 3^{\nqubits - |S|} \cdot 3^{\nqubits - |S'|}} &= \frac{1}{10} \sum_{P \in \{X,Y,Z\}^{\otimes \nqubits}} \left(\sum_{\underset{Q \neq \eye_d}{Q \triangleleft P}} \frac{\alpha_Q^2}{3^{\nqubits - w(Q)}}\right)^2  \\
    &\leq \frac{1}{10} \left(\sum_{P \in \{X,Y,Z\}^{\otimes \nqubits}} \sum_{\underset{Q \neq \eye_d}{Q \triangleleft P}} \frac{\alpha_Q^2}{3^{\nqubits - w(Q)}}\right)^2 \\
    &= \frac{\|\alpha_Q\|_2^4}{10}
\end{align*}
The final covariance term requires some careful counting. We state the result here and defer the proof to~\cref{sec:ub-det-proof}.
\begin{lemma}[Covariance of Deterministic Pauli Mixedness Testing] \label{lem:cov-pauli}
    Given $\rho$ with Pauli Decomposition $\{\alpha_Q\}_{Q \in \mathcal{Q}}$,
    \begin{align*}
        \sum_{P \in \{X,Y,Z\}^{\otimes \nqubits}} \sum_{\emptyset \subset S,S' \subseteq [\nqubits]} \frac{\alpha_P(S \Delta S')^2}{3^{\nqubits - |S|} \cdot 3^{\nqubits - |S'|}} &= \left(\frac{10}{3}\right)^\nqubits +\left(\frac{10}{3}\right)^\nqubits \sum_{Q \in \mathcal{Q}} \frac{\alpha_Q^2}{5^{w(Q)}} \\
        &\leq  \left(\frac{10}{3}\right)^\nqubits + \left(\frac{10}{3}\right)^\nqubits \|\alpha_Q\|_2^2
    \end{align*}
\end{lemma}
\noindent The variance can now be bounded in terms of the norm of the Pauli coefficients.
\begin{align*}
    \Var(W) \leq \frac{10^\nqubits}{3^\nqubits n^2} + \frac{2 \|\alpha_Q\|_2^2}{n} + \frac{11 \cdot 10^{\nqubits} \|\alpha_Q\|_2^2}{3^\nqubits n^2} + \frac{11 \cdot 10^{\nqubits}}{3^\nqubits n^2}+ \frac{\|\alpha_Q\|_2^4}{10}.
\end{align*}
With $n = \frac{200}{\eps^2} \cdot \sqrt{\frac{10}{3}}^\nqubits$,
\begin{align*}
    \Var(W) &\leq \frac{\eps^4}{400} + \frac{\eps^2 \|\alpha_Q\|_2^2}{100} + \frac{\eps^2 \|\alpha_Q\|_2^2 }{100} + \frac{\eps^4}{400} + \frac{\|\alpha_Q\|_2^4}{10} \\
    &\leq \frac{ \|\alpha_Q\|_2^4}{5}
\end{align*}
under the alternative hypothesis. We notice that $\Var(W) \leq \frac{\eps^4}{200}$ when $\rho = \qmm$.
By Chebyshev's Inequality,
\begin{align*}
\Pr_{\rho \neq \qmm}\left[W \geq \frac{\|\alpha_Q\|_2^2}{5} \geq \frac{\eps^2}{5}\right] \geq \frac{2}{3}, \quad \Pr_{\rho = \qmm}\left[W < \frac{\eps^2}{5}\right] \geq \frac{2}{3}.
\end{align*}
With black box amplification~\cite[Lemma 1.1]{canonne2022topics}, the simple fixed scheme performs mixedness testing with $\bigO{\frac{\sqrt{30}^\nqubits}{\eps^2} \log \frac{1}{\delta}}$ copies.
\end{proof}
\subsection{Randomized Pauli Measurements} \label{sec:ub-rand}
We will make use of randomness to improve upon simple tester and show that it is optimal up to $poly(N)$ factors. To understand the main ideas of our protocol, we break the mixedness tester into two parts.
\begin{algorithm}[h]
\caption{Pauli Mixedness Tester}\label{alg:rand_test}
    \begin{algorithmic}[1] % [1] adds line numbers
    \Require Copy access to $\sigma^{\otimes n}$, $\eps \in (0,1)$
    \Ensure{
    $\begin{cases} 
    \text{\texttt{Maximally-Mixed}} & \text{w.p. } \frac{2}{3} \quad \text{if } \sigma = \rho_{\text{mm}} \\ 
    \text{\texttt{Not Maximally-Mixed}} & \text{w.p. } \frac{2}{3} \quad \text{if } \|\sigma - \rho_{\text{mm}}\|_1 \geq \varepsilon 
    \end{cases}$} 
    \For{$i=1,2, \ldots, N$}
        \State $\delta_i \gets 100^{-(i+N)}$, $l_i \gets \max\left\{1,\lfloor\log_3 \frac{3^{i}}{9N}\rfloor\right\}$
        \State $\mathcal{B}_i \gets 3^i \cdot N \cdot 100$ i.i.d samples from $\unif[\{X,Y,Z\}^{\otimes N}]$ 
        \For{$P \in \mathcal{B}_i$}
            \State  $r_P \gets$ measurement outcome distribution of measuring $\sigma$ in Pauli basis $P$ 
            \For{$w=l_i,l_i+1, \ldots N$}
                \State $\beta^{(w)}_{i} \gets \frac{3^i}{9N 3^w}$
                \State $T_w \gets \mathtt{KrawtCheck}_{r_P}\left(k=w, \eps^2, \beta^{(w)}_{i}, \delta_i \right)$
                \If{$T_w = \mathtt{YES}$ }
                    \State \Return $\mathtt{Not\;Maximally-Mixed}$
                \EndIf
            \EndFor
        \EndFor
    \EndFor 
    \State \Return $\mathtt{Maximally-Mixed}$
    \end{algorithmic}
\end{algorithm}
\paragraph{Influence detection.} The two higher-level for-loops perform a random search on the Pauli basis $\{X,Y,Z\}^{\otimes \nqubits}$ to find a $P$ with a large Pauli influence $L_P$, where~\cref{lem:pauli_work_invest} shows such a basis can be found through Levin's work investment strategy. At a high-level, the work strategy states that either there exists a large number of Pauli measurements that have a typical influence or a sparser set of Pauli measurements with high influence. Thus, we can sample many Pauli bases to hopefully detect the highly influential Pauli measurement or sample fewer Pauli bases to detect a more average Pauli influence. The proof of this statement is provided in~\cref{sec:ub-proof}.
\begin{algorithm}[h]
\caption{KrawtCheck}\label{alg:sparse}
    \begin{algorithmic}[1] % [1] adds line numbers
    \Require  Sample access to distribution $p$ that is either $(\Delta, k, \beta)$-correlated or uniform for $\Delta, \beta,\delta \in (0,1]$, $k \in [N]$.
    \Ensure{
     $\begin{cases} 
    \text{\texttt{YES}} & \text{w.p. } 1-\delta \quad \text{if } p \text{ is $(\Delta, \beta, k)$-correlated} \\
    \text{\texttt{NO}} & \text{w.p. } 1-\delta \quad \text{if } p = u 
    \end{cases}$
    }
    \State $n \gets \frac{400\sqrt{2 \binom{N}{k}}}{\beta\Delta} $ \\
    \State Sample $X_1, \ldots, X_n \sim p$ and $Y_1, \ldots, Y_n \sim p$ 
    \State $W_k \gets \frac{1}{n^2} \sum_{i=1}^n \sum_{j=1}^n \krawt{k}{d_{H}(X_i,Y_j)}$
    \If{$W_k \geq \frac{\beta \Delta}{50}$} 
    \State \Return $\mathtt{YES}$ 
    \Else 
    \State  \Return $\mathtt{NO}$
    \EndIf 
    \end{algorithmic}
\end{algorithm}
\paragraph{Detection based off weight.} Once an influential $P$ is sampled, the inner for-loop performs a weight-based check on each weight according to the influence factor $3^{w}$, which is embedded in the statistic $L_P$. We know that one of the weights must dominate $L_P$, resulting in a certain weight norm that dominates norm of the other weights by factor $\beta = \tildeOmega{3^{i-w}}$ (if the influence threshold was $3^i$). The check for the dominant weight is then performed with the \verb|KrawtCheck| subroutine, using the idea of correlation-concentrated testing shown by~\cref{thm:sparse-test} and proven in~\cref{sec:ub-proof}. Thus, the complexity of detection will scale accordingly by factor $3^{-i} \sqrt{9^w \binom{\nqubits}{w}} \cdot 3^i$, which is at most $\sqrt{10}^N$ at $w=9\nqubits/10$ (forgoing the $poly(N)$ factors that come from grouping $L_P$ by weight and union bound of the error probabilities). We remark that $9\nqubits/10$ matches the minimum weight of the Pauli observables used for the lower bound construction.

We will now prove~\cref{thm:pauli-ub} by showing proving the correctness of~\cref{alg:rand_test} and showing that it uses $\tildeO{\frac{\sqrt{10}^\nqubits}{\eps^2}}$ copies.
\begin{proof}
We start out by showing that~\cref{alg:rand_test} uses $\tildeO{\frac{\sqrt{10}^N}{\eps^2}}$ copies. By~\cref{thm:sparse-test},
\begin{align*}
    n = \sum_{i=1}^N \bigO{3^{i} \cdot N} \sum_{w=l_i}^N \bigO{\frac{ \sqrt{ 9^w\binom{N}{w}} N^2}{3^i \Delta}} = \bigO{\frac{\sqrt{10}^N N^5}{\eps^2}} = \tildeO{\frac{\sqrt{10}^N}{\eps^2}}.
\end{align*}
    We proceed to show correctness under the null hypothesis. When $\sigma = \qmm$, the outcome distribution is uniform under any measurement, so the error probability comes from  \verb|KrawtCheck| incorrectly outputting $\mathtt{YES}$ when the given distribution is $r_P = u$. By union bound,
\begin{align*}
    \Pr_{\qmm}[T=\mathtt{Not\;Maximally-Mixed}] &\leq \sum_{i=1}^N 100 \cdot 3^i \cdot N^2 \cdot \delta_i  \leq \sum_{i=1}^N \frac{100 \cdot 3^i \cdot N^2}{100^{i+N}}  \\
    &\leq \sum_{i=1}^N 30^{-i} < \sum_{i=1}^{\infty} 30^{-i} <\frac{1}{4}.
\end{align*}
Now, we prove correctness when $\sigma \neq \qmm$. From \cref{pre:dist}, it can be seen that
\begin{align*}
L \eqdef \|\sigma-\qmm\|_2^2 = \sum_{Q \in \mathcal{Q}} \alpha_Q^2 \geq \eps^2
\end{align*}
when $\|\sigma - \qmm\|_1 \geq \eps$. We will invoke \cref{lem:pauli_work_invest} to guarantee that a detectable basis measurement can be sampled with (reasonably) high probability. There exists $i^* \in [N]$ such that
\begin{align*}
    \Pr[L_P \geq 3^{i^*-2} L] \geq \frac{3^{-i^*}}{N}.
\end{align*}
Consider breaking up Pauli influence $L_P$ by weight
\begin{align*}
    L_P = \sum_{Q \triangleleft P} 3^{w(Q)} \alpha_Q^2 = \sum_{w=1}^N 3^{w} \sum_{S \in \binom{[N]}{w}} \alpha_{P}(S)^2 = \sum_{w=1}^N L_P^{(w)},
\end{align*}
where we define $L_P^{(w)} \eqdef 3^{w} \sum_{S \in \binom{[N]}{w}} \alpha_{P}(S)^2 = 3^{w} \|\alpha_{P}^{=w}\|_2^2$. If we take $w^* \eqdef \arg \max_{w \in [N]} L_P^{(w)}$, then
\begin{align*}
    \|\alpha_{P}^{=w^*}\|_2^2
    \geq \frac{3^{i^*-2} L}{3^{w^*} N} \geq \beta_{i^*}^{(w^*)} \eps^2, \quad
     \|\alpha_{P}^{=w^*}\|_2^2 \geq \frac{3^{i^*-2} L}{3^{w^*} N} \geq \beta_{i^*}^{(w^*)}\|\alpha_{P}\|_2^2
\end{align*}
for $L_P \geq 3^{i*-2} L$. We note that $w^* \geq l_{i^*}$. Otherwise, we obtain a contradiction $\|\alpha_{P}^{=w^*}\|_2^2 > \|\alpha_{P}\|_2^2$. Using~\cref{obs-pauli-fourier},
 \begin{align*}
     r_P(x) &= \frac{1}{2^N} + \frac{1}{2^N} \sum_{\emptyset \subset S \subseteq N} \parity{x} \alpha_{P}(S),
 \end{align*}
 we see that $r_P$ is $\left(\eps^2, w^*, \beta_{i^*}^{(w^*)}\right)$-correlated. The test could only fail if either (1) A desirable $P$ is not sampled in iteration $i^*$ or (2) \verb|KrawtCheck| returns $\mathtt{NO}$ on the outcome distribution of a desired $P$. We know that \verb|KrawtCheck| should return $\mathtt{NO}$ with probability $\leq \delta_{i^*}$ when $r_P$ is $\left(\eps^2, w^*, \beta_{i^*}^{(w^*)}\right)$-correlated. Thus, by union bound,
 \begin{align*}
     \Pr_{\sigma \neq \qmm}[T(x^\nsamp) &= \mathtt{Maximally-Mixed}]  \leq \Pr\left[\max_{P \in \mathcal{B}_{i^*}} L_P < 3^{i^* -2} L \right] + \delta_{i^*}. \\
 \end{align*}
 We notice that $Z \eqdef \sum_{P \in \mathcal{B}_{i^*}} \indic{L_P \geq 3^{i^*-2} L} \sim Bin(J, \gamma)$ where $J = 3^{i^*} \cdot N \cdot 100$ and $\gamma \geq \frac{3^{-i^*}}{N}$, resulting in 
 \begin{align*}
     \Pr\left[\max_{P \in \mathcal{B}_{i^*}} L_P < 3^{i^* -2} L \right] = (1-\gamma)^J \leq e^{-\gamma J} \leq e^{-100} < \frac{1}{8}.
 \end{align*}
 The error probability under $\sigma\neq\qmm$ is
 \begin{align*}
      \Pr_{\sigma \neq \qmm}[T(x^n) = \mathtt{Maximally-Mixed}] < \frac{1}{8} + \delta_{i^*} < \frac{1}{4}.
 \end{align*}
With black box amplification, we have proven~\cref{thm:pauli-ub}.
\end{proof}

\section{Mixedness Testing Lower Bounds} \label{sec:gen-bound}
The key idea behind our nearly-optimal lower bound is to derive bounds that depend on the measurement constraints. This allows us to obtain tighter guarantees under restricted measurement settings. Recall, the goal is to design a measurement scheme $\POVM^n= \POVM^{(1)} \otimes \POVM^{(2)} \otimes \ldots \otimes \POVM^{(\nsamp)}$ such that 
\begin{align} \label{def:cert-goal}
    \Pr_\qmm[T(x^n) = \mathtt{Maximally-Mixed}] \geq \frac{2}{3}, \quad
    \min_{\sigma: \tracenorm{\sigma-\qmm} \geq \eps}\Pr_{\sigma}[T(x^n) = \mathtt{Not \; Maximally-Mixed}] \geq \frac{2}{3},
\end{align}
where $x^n$ are the outcomes of the adaptive measurement scheme $\POVM^{(i)} = g_i(x^{i-1})$ for some randomized function $g_i: \mathcal{X}^{i-1} \rightarrow \povmset$. For single copy measurements, it suffices to consider the conditional distributions of the measurement outcomes $x_i | x^{i-1}$.

\subsection{Our Measurement-dependent Framework}
To establish the measurement-dependent lower bound, we will be using the ``tried and true'' Lecam's two point method. Consider an ensemble of states $\mathcal{P} \subset \{\sigma: \tracenorm{\sigma -\qmm } \geq \eps\}$ where each state is drawn according to distribution $\nu$. 

In classical identity testing, Lecam's method states that a valid tester must be able to distinguish between the mixture of the ensemble of perturbed distributions and the null hypothesis~\cite{Paninski08}. Similarly, a valid quantum certification protocol must be able to distinguish $\qmm$ from $\expectDistrOf{\rho \sim \nu}{\rho}$.
\begin{lemma}[{\cite[Lemma 3.3]{Liu2024role}}]
Suppose nature selects $\rho = \qmm$ with probability $1/2$ and draws from the ensemble $\rho \sim \nu$ with probability $1/2$. Then, a measurement scheme that performs mixedness testing must satisfy
\begin{align} \label{eq:lecam}
        \frac{1}{2}\le \kldiv{\q^{\out^\ns}}{\p^{\out^\ns}}=\sum_{i=0}^{\ns-1}\expectDistrOf{\q^{\out^\ns}}{\kldiv{\q^{\out_{i+1}|\out^{i}}}{\p^{\out_{i+1}|\out^{i}}}},
\end{align}
where $\q^{\out^\nsamp}(x^n) = \expectDistrOf{\sigma \sim \nu}{\Pr[\forall_{i \in [\nsamp]} \; \POVM^{(i)}(\sigma) = x_i]}$ and $\p^{\out^\nsamp}(x^n) = \Pr[\forall_{i \in [\nsamp]} \; \POVM^{(i)}(\qmm) = x_i]$. The corresponding conditional distributions are denoted by $\p^{x_{i+1}|x^i}$ and $\q^{x_{i+1}|x^i}$.
\end{lemma}
\paragraph{Hard instance construction.}
To provide a lower bound on $\nsamp$, it suffices to bound the expected KL-divergence between the outcome distributions of the maximally mixed state and the mixture of ensembles in terms of $\nsamp$ and the measurement constraints. To minimize the KL-divergence, we need to pick an ensemble of hard instance states that are the least distinguishable from $\qmm$ under the measurement.

We consider the following general hard-instance construction used in recent works of tomography and certification ~\cite{acharya2025single, ADLY2025Paulinot, Liu2024role}.
\begin{definition}[Hard Instance Construction] \label{def:hard-instance}
    Let $d^{3/2} \leq \ell \leq d^2-1$ and $\{V_1,V_2, \ldots, V_{d^2-1}, V_{d^2 } = \frac{\eye_d}{d}\}$ be an orthonormal basis of $\D(\mathbb{C}^d)$. Let $z \sim \nu = Rad(\frac{1}{2})^{\otimes \ell}$, then the ensemble of states $\sigma_z$ is defined as
   \begin{align*}
       \sigma_z \eqdef \qmm + \barDelta_z, \quad \Delta_z \eqdef \frac{c\eps}{\sqrt{d}} \cdot \frac{1}{\sqrt{\ell}} \sum_{m=1}^\ell z_m V_m, \quad \barDelta_z \eqdef \Delta_z \min\left\{1, \frac{1}{2d \opnorm{\Delta_z}}\right\}.
   \end{align*} 
\end{definition}
One can think of this construction as the quantum analogue to the Paninski construction~\cite{Paninski08}. We perturb $\qmm$ along uniformly random orthogonal directions on a $\eps$-sized hypercube. $\Delta_z$ incurs normalization to ensure positive semi-definiteness of each state. We provide a useful fact regarding the normalization factor.
\begin{theorem}[{\cite[Theorem 5.2]{acharya2025single}}]
Consider the construction in~\cref{def:hard-instance} when $c \leq \frac{1}{200}$. Let $\mathcal{Z} \eqdef \{z | z \in \{-1,1\}^\ell, \Delta_z \neq \barDelta_z\}$. Then,
\begin{align*}
   \Pr[z \in \mathcal{Z}] < d\exp(-\ell^{\frac{1}{4}}).
\end{align*}
\end{theorem}
Essentially, the normalization factor is negligible as incurring the ``bad states'' happens with very low probability. We proceed to prove~\cref{thm:gen-lb}.

\paragraph{Bridging the gap with average mutual information.}
 We prove a key result that reduces mixedness testing to the problem of communicating $\ell$ bits over a measurement-induced noisy channel.
 
{\begin{lemma}[Distinguishability from Maximally Mixed State]
Given the state ensemble defined in~\cref{def:hard-instance},
    \begin{align*}
      \expectDistrOf{\q^{\out^\ns}}{\kldiv{\q^{\out_{i+1}|\out^{i}}}{\p^{\out_{i+1}|\out^{i}}}} \leq \frac{ c^2 \ln 4 \; \eps^2}{\ell} \opnorm{\Luders_\povmset} \sum_{m=1}^\ell \mutualinfo{z_m}{x^i}
    \end{align*}
    for $i=1,2,\ldots,\nsamp$.
\end{lemma}
\begin{proof}
    We start by bounding the KL-divergence with the $\chi^2$-divergence to obtain a better dependence on the measurement constraints. 
    \begin{align}
        \expectDistrOf{\q^{\out^\ns}}{\kldiv{\q^{\out_{i+1}|\out^{i}}}{\p^{\out_{i+1}|\out^{i}}}} &\leq  \expectDistrOf{\q^{\out^\ns}}{\chisquare{\q^{\out_{i+1}|\out^{i}}}{\p^{\out_{i+1}|\out^{i}}}} \\
        &= \expectDistrOf{\q^{\out^\ns}}{\expectDistrOf{X \sim \p^{x_{i+1}|x^i}}{\beta_i(X)^2}}  \label{eq:chi-square},
    \end{align}
    where $\beta_i(x) = \frac{q^{x_{i+1}|x^i} (x) - p^{x_{i+1}|x^i}(x)}{p^{x_{i+1}|x^i}(x)}$. We use Born's rule to obtain the following expression for $\beta_i(x)^2$,
    \begin{align*}
        \beta_i(x)^2 &= \left(\frac{\expectDistrOf{Z|x^i}{ \Tr(M^{(i)}_x\sigma_Z) - \Tr(M^{(i)}_x \qmm)}}{\Tr\left(M^{(i)}_x \qmm\right)}\right)^2 =  \frac{\Tr\left(M^{(i)}_x \expectDistrOf{Z|x^i}{ \barDelta_Z}\right)^2}{\Tr\left(M^{(i)}_x \qmm\right)^2}\\
        &=  \frac{\dotprodH{M^{(i)}_x}{\expectDistrOf{Z|x_i}{\barDelta_Z}}^2}{\Tr\left(M^{(i)}_x \qmm\right)^2},
    \end{align*}
    where the last line used linearity of expectation and inner product. We will now unravel the expectation over the $\beta_i(X)$'s,
    \begin{align*}
        \expectDistrOf{X \sim \p^{x_{i+1}|x^i}}{\beta_i(X)^2} &= \sum_{x \in \mathcal{X}} \frac{\dotprodH{M^{(i)}_x}{\expectDistrOf{Z|x_i}{\barDelta_Z}}^2}{\Tr\left(M^{(i)}_x \qmm\right)} = d \sum_{x \in \mathcal{X}} \frac{\dotprodH{M^{(i)}_x}{\expectDistrOf{Z|x_i}{\barDelta_Z}}^2}{\Tr\left(M^{(i)}_x\right)} \\
        &\leq d \sum_{x \in \mathcal{X}} \frac{\dotprodH{M^{(i)}_x}{\expectDistrOf{Z|x_i}{\Delta_Z}}^2}{\Tr\left(M^{(i)}_x\right)}= d \dotprodH{\Luders_{\POVM^{(i)}}(\expectDistrOf{Z|x^{i}}{\Delta_Z})}{\expectDistrOf{Z|x^{i}}{\Delta_Z}} \\
        &\leq d \cdot \opnorm{\Luders_\povmset}\dotprodH{\expectDistrOf{Z|x^{i}}{\Delta_Z}}{\expectDistrOf{Z|x^{i}}{\Delta_Z}}.
    \end{align*}
    It now suffices to express the expected norm of the perturbation.
    \begin{align*}
        \dotprodH{\expectDistrOf{Z|x^{i}}{\Delta_Z}}{\expectDistrOf{Z|x^{i}}{\Delta_Z}} &= \frac{c^2 \eps^2}{\ell d}\sum_{m,l \in [\ell]} \expectDistrOf{z_m | x^{i}}{z_m} \expectDistrOf{z_l | x^{i}}{z_l} \dotprodH{V_m}{V_l} \\
        &= \frac{c^2 \eps^2}{\ell d} \sum_{m,l \in [\ell]} \expectDistrOf{z_m | x^{i}}{z_m} \expectDistrOf{z_l | x^{i}}{z_l} \delta_{ml} \\
        &= \frac{c^2 \eps^2}{\ell d} \sum_{m=1}^\ell \expectDistrOf{z_m | x^{i}}{z_m}^2.
    \end{align*}
    We invoke a useful lemma to characterize the expectation of the random variable $\expectDistrOf{z_m | x^{i}}{z_m}^2$ when $z_m$ is restricted to a binary random variable.  
    \begin{lemma}[{\cite[Lemma 6]{ACLST22iiuic}}]
    \label{lemma:MI-channel-lower}
    Let  $Z\in\{-1, 1\}^\ell$ be drawn uniformly and $Y$ be a random variable. Then for each $m\in[\ell]$,
    \[
     \expect{\expectCond{Z_m}{Y}^2} \leq \ln 4 \cdot \mutualinfo{Z_m}{Y}.
    \]
    \end{lemma}
    \noindent Substituting what we have into~\cref{eq:chi-square},
    \begin{align*}
        \expectDistrOf{\q^{\out^\ns}}{\kldiv{\q^{\out_{i+1}|\out^{i}}}{\p^{\out_{i+1}|\out^{i}}}} &\leq \frac{c^2 \eps^2}{\ell} \opnorm{\Luders_\povmset} \sum_{m =1}^\ell \expectDistrOf{\q^{\out^\ns}}{\expectDistrOf{z_m | x^{i}}{z_m}^2} \\
        &\leq \frac{c^2 \ln 4 \; \eps^2}{\ell}  \opnorm{\Luders_\povmset} \sum_{m =1}^\ell \mutualinfo{z_m}{x^i}.\qedhere
    \end{align*}
    \end{proof}
    \noindent We will now characterize the average mutual information in terms of the measurement constraints to conclude our lower bound framework for adaptive single-copy state certification.
    \begin{theorem}[{\cite[Theorem 4.4]{ADLY2025Paulinot}}]
    \label{thm:avg-MI-upper}
    Let $x^i$ be the $x_1, x_2,\ldots,x_i$ marginal of the outcome distribution of performing $\POVM^{(1)} \otimes \POVM^{(2)} \otimes \ldots \otimes \POVM^{(n)}$ on the random state $\sigma_z^{\otimes n}$. Then, for all $i\in[\ns]$,
    \begin{align}
        \frac{1}{\ell} \sum_{m=1}^\ell \mutualinfo{z_m}{x^i} \leq \frac{8ic^2 \eps^2}{\ell^2}  \sum_{m=1}^\ell \dotprodH{\Luders_\POVM(V_m)}{V_m} + 16ic^2\eps^2 \Pr[z \in \mathcal{Z}].
    \end{align}
    \end{theorem}

    \noindent Putting it all together, we have
    \begin{align} \label{eq:mdep-lb}
        \frac{1}{2}\le \kldiv{\q^{\out^\ns}}{\p^{\out^\ns}} &\leq \frac{16 \; c^4 \nsamp^2 \eps^4}{\ell^2}   \opnorm{\Luders_\povmset} \cdot \sum_{m=1}^\ell \dotprodH{\Luders_\POVM(V_m)}{V_m} + 32 \; n^2c^4\eps^4 d \exp\{-\ell^{1/4}\}. 
    \end{align}
    We can select a basis $\{V_1,V_2, \ldots ,V_\ell\}$ based on $\povmset$ to obtain a measurement dependent lower bound for adaptive state certification. Picking a complete orthonormal basis for $\mathcal{V}$ yields us the following measurement-dependent lower bound,
    \begin{align*}
        \frac{1}{2} &\leq \frac{8 \ln 4 \; c^4 \nsamp^2 \eps^4}{d^4}   \opnorm{\Luders_\povmset} \|\Luders_\povmset\|_1 + 16 \ln 4 \; n^2c^4\eps^4 d \exp\{-\ell^{1/4}\}
    \end{align*}
    which proves~\cref{thm:gen-lb} since $d\exp\{-\ell^{1/4}\} \leq \frac{\opnorm{\Luders_\povmset} \|\Luders_\povmset\|_1}{d^4}$ for $\nqubits \geq T$, where $T$ is a sufficiently large constant.

\subsection{Hard-instance Construction for Single-Qubit Measurements}
We now revisit the task of single-qubit certification and establish a lower bound in this setting using the results from~\cref{sec:gen-bound}. We will achieve the desired lower bound by adopting the construction used for Pauli tomography~\cite[Section 6.3]{acharya2025single}. We restate the results of this construction and apply it to our lower bound framework under Pauli basis measurements. Then, we will extend our lower bound to single-qubit protocols, which will prove~\cref{thm:sq-lb}.
\begin{proof}
An important observation is that the Pauli observables $Q \in \mathcal{Q}$ serve as the eigenbasis of the Pauli MIC.
\begin{theorem}[{\cite[Lemma 5.2]{ADLY2025Paulinot}}] Let $\POVM_P$ be the measurement along the Pauli basis $P \in \{X,Y,Z\}^N$, then for $Q \in \mathcal{Q}$,
\
\begin{align*}
    \Luders_{\POVM_P}(Q) = \indic{Q \triangleleft P} Q.
\end{align*}
\end{theorem}
\noindent We observe that there exists exactly $3^{N-w(Q)}$ possible Pauli bases where the eigenvalue of $Q$ is non-zero. Thus, the hardness of the construction relies solely on the weight of the Pauli observables chosen. Consider the construction with orthonormal directions $\mathcal{V} = \left\{\frac{Q}{\sqrt{2}^N} |Q \in \mathcal{Q} \land w(Q) \geq w^*\right\}$, then for any $P \in \{X,Y,Z\}^{\otimes \nqubits}$
\begin{align*}
   \sum_{m=1}^\ell \dotprodH{\Luders_{\POVM_P}(V_m)}{V_m} = \sum_{w=w^*}^\nqubits \binom{\nqubits}{w}, \quad \ell =  \sum_{w=w^*}^\nqubits 3^w \binom{\nqubits}{w} \geq 3^{w^*} \sum_{w=w^*}^\nqubits \binom{\nqubits}{w}
\end{align*}
The error analysis of the upper bound for Pauli tomography~\cite[Section 2.2]{acharya2025single} suggests that the optimal weight to use is $w^* = \lceil 9\nqubits/10\rceil$. As we are minimizing the same quantity, we will use this same intuition for certification. We can show the given weight guarantees $\ell \geq d^{3/2}$ for $N \geq 4$:
\begin{align*}
    \ell \geq \frac{3^{\lceil 9N/10 \rceil} 2^{\nqubits h(1/10)}}{\sqrt{2N}}  \geq \frac{2^{1.88 N}}{\sqrt{2N}}       \geq 2^{-0.38N}2^{1.88N} = \dims^{3/2},
\end{align*}
where Stirling's approximation for binomial coefficients~\cite[Lemma 8]{macwilliams1978theory} was used.
We bound the spectral capacity of the MIC with the same approximation,
\begin{align*}
    \frac{\sum_{m=1}^\ell \dotprodH{\Luders_{\POVM_{P}}(V_m)}{V_m} }{\ell^2} \leq \frac{1}{9^{9\nqubits/10} \sum_{w=\lceil 9\nqubits/10\rceil}^\nqubits \binom{\nqubits}{w}} \leq \frac{\sqrt{2N}}{9^{9\nqubits/10} 2^{\nqubits h(9/10)}} = \frac{\sqrt{2\nqubits}}{10^\nqubits}.
\end{align*}
Observe that $\Pr[z \in \mathcal{Z}] \leq \frac{\sqrt{\nqubits}}{10^\nqubits}$ for $\nqubits \geq 10$. Substituting this result into~\cref{eq:mdep-lb},
\begin{align*}
    \frac{1}{2}\le \kldiv{\q^{\out^\ns}}{\p^{\out^\ns}} = \bigO{\frac{\eps^4 n^2 \sqrt{\nqubits}}{10^\nqubits}},
\end{align*}
 which yields the desired lower bound $\nsamp = \bigOmega{\frac{\sqrt{10}^\nqubits}{N^{1/4} \eps^2}} = \tildeOmega{\frac{\sqrt{10}^\nqubits}{\eps^2}}$ for adaptive Pauli measurements.
\paragraph{From Pauli to Single-Qubit Measurements. } It is straightforward to show that the lower bound even applies for arbitrary single-qubit measurements using the fact that Pauli observables form a unitary 2-design for $\mathcal{D}(\mathbb{C}^2)$. We restate the result proven in~\cite{acharya2025single}.
\begin{lemma}[{\cite[Lemma 6.5]{acharya2025single}}]
    Let $P_1, \ldots, P_{\ell}$ be the normalized Pauli observables of weight at least $w*$. Then for any single-qubit measurement $\POVM$,
    \begin{align*}
        \sum_{m=1}^\ell \dotprodH{P_m}{\Luders_{\POVM}(P_m)} \leq \sum_{w=w^*}^N \binom{N}{w}.
    \end{align*}
\end{lemma}
\noindent Thus, we have proven~\cref{thm:sq-lb}.
\end{proof}
\subsection{Deterministic Pauli Mixedness Testing} \label{sec:lb-fixed}
We will now use the lower bound framework for deterministic state certification~\cite{Liu2024role} to establish the lower bound for fixed Pauli measurements seen in~\cref{thm:det-pauli}. 
\begin{proof}[Proof of lower bound in~\cref{thm:det-pauli}]
    
By~\cite[Theorem 4.7]{Liu2024role}, we have
\begin{align*}
    \frac{1}{2} \leq \max_{\POVM^n} \min_{\mathcal{V}} \; \exp\left\{\frac{cn^2\eps^4}{\ell^2}  \sum_{m=1}^\ell \dotprodH{V_m}{\bar{\Luders}(V_m)}^2\right\} - 1 + 2 d \exp\left\{-\ell^{1/4}\right\}
\end{align*}
where $\bar{\Luders} = \frac{1}{n} \Luders_{\POVM_i}$ and $\mathcal{V}$ is chosen to be the eigenvectors of $\bar{\Luders}$. We also highlight on subtlety here. The last term in the inequality is the new concentration that results from tuning the lower bound construction to accept smaller $\ell$ while ensuring that $\sigma_z$ is a valid state with high probability~\cite[Theorem 5.2]{acharya2025single}. The concentration used in the proof of~\cite[Theorem 4.7]{Liu2024role} can be easily exchanged with this new concentration result and allows us to perturb the state along much fewer Pauli observables.

First, we make the following observation. Since there are $3^\nqubits$ possible Pauli measurements, we can rewrite $\bar{\Luders}$ as the following,
\begin{align*}
    \bar{\Luders} = \frac{1}{n} \sum_{P \in \{X,Y,Z\}^{\otimes \nqubits}} \{\text{\# of allocated copies for $P$}\} \cdot \Luders_{\POVM_P} = \sum_{P \in \{X,Y,Z\}^{\otimes \nqubits}} \nu(P) \cdot \Luders_{\POVM_P},
\end{align*}
where $\nu(P) \eqdef \frac{\{\text{\# of allocated copies for $P$}\}}{n}$ is a probability measure over the measurement bases. Now, we consider lower bound constructions over normalized full weight observables: $\mathcal{V} \subseteq \left\{X,Y,Z\right\}^{\otimes \nqubits}$ and $\ell = |\mathcal{V}|$. By linearity of inner product, we can write the inner product under Pauli observables as
\begin{align*}
    \dotprodH{Q}{\bar{\Luders}(Q)} = \nu(Q)
\end{align*}
Thus, the quantity that we are trying to minimize is 
\begin{align*}
    \min_{\mathcal{V} \subseteq \{X,Y,Z\}^{\otimes \nqubits}}  \frac{1}{\ell^2}\sum_{Q \in \mathcal{V}}  \nu(Q)^2.
\end{align*}
Sort the bases in increasing order of the masses according to $\nu$ and call them $P_1, P_2, \ldots, P_{3^\nqubits}$. Consider selecting $\ell =\lfloor3^\nqubits/2\rfloor \geq 3^{\nqubits}/4$ full weight observables that correspond to the bases $P_1, P_2, \ldots, P_{\lfloor3^\nqubits/2\rfloor}$. This is a valid construction since $\ell \geq d^{3/2}$ when $\nqubits \geq 25$. Thus,
\begin{align*}
 \min_{\mathcal{V} \subseteq \{X,Y,Z\}^{\otimes \nqubits}}  \frac{1}{\ell^2}\sum_{Q \in \mathcal{V}}  \nu(Q)^2 \leq \frac{16}{9^\nqubits} \sum_{i=1}^{3^\nqubits/2} \nu(P_i)^2
\end{align*}
We note that $\nu(P_{\lfloor3^{\nqubits}/2\rfloor}) \leq \frac{2}{3^\nqubits}$. Otherwise, $\sum_{i=\lfloor3^\nqubits/2\rfloor}^{3^\nqubits} \nu(P_i) > 3^\nqubits/2 \cdot 2/3^\nqubits = 1$, resulting in a invalid probability distribution. Therefore, the $\ell_2$ norm of the smaller half the probability masses can be bounded.
\begin{align*}
    \min_{\mathcal{V} \subseteq \{X,Y,Z\}^{\otimes \nqubits}}  \frac{1}{\ell^2}\sum_{Q \in \mathcal{V}}  \nu(Q)^2 \leq \frac{32}{27^\nqubits} \sum_{i=1 }^{\lfloor3^{\nqubits}/2\rfloor} \nu(P_i) \leq \frac{32}{27^\nqubits}.
\end{align*}
Since $\exp\{\cdot\}$ is a monotonic function, we can substitute this result in the final bound to obtain
\begin{align*}
    \frac{1}{2} \leq  \exp\left\{\frac{32cn^2\eps^4}{27^\nqubits}\right\} - 1 + 2 d \exp\left\{-\ell^{1/4}\right\}.
\end{align*}
When $\nqubits \geq 25$, $\exp\left\{-\ell^{1/4}\right\} \leq \exp\left\{\frac{32cn^2\eps^4}{27^\nqubits}\right\}/8d$. Resulting in 
\begin{align*}
\frac{1}{2} \leq \frac{5}{4} \exp\left\{\frac{32cn^2\eps^4}{27^\nqubits}\right\} - 1.
\end{align*}
Re-arranging both sides, we have $n = \bigOmega{\frac{\sqrt{27}^\nqubits}{\eps^2}}$, proving the lower bound in~\cref{thm:det-pauli}.
\end{proof}
\section*{Conclusion}
In this paper, we have nearly resolved the copy complexity of single-qubit mixedness testing and established a measurement-dependent lower bound framework for certification under adaptive single-copy measurement protocols. We believe that there are several interesting open questions in single-qubit inference, which we will highlight below.
\paragraph{Arbitrary Target States.} In this paper, we addressed the task of mixedness testing, which is the special case of certification when the target state is the maximally mixed state. In particular, the algorithm in our upper bound is heavily based on the fact that the outcome distribution of the target state remains uniform. Unlike classical identity and uniformity testing, certification upper bounds with alternative target states tend to be a bit more involved due to the non-commuting nature of states~\cite{ChenLO22instance}. It would be interesting to see what additional algorithmic techniques should be used for single-qubit state certification.
\paragraph{Structural Assumptions.} Oftentimes, it is more practical to consider low-dimensional state structure as relevant states are often pure or near pure. Thus, another noteworthy line of work would be to consider single-qubit certification under structural assumptions such as rank or even establishing instance-optimal bounds dependent on the target state. 

\section*{Acknowledgments}
JA, AD and NY are supported by the National Science Foundation under Grant No. CCF-2553759. YL is supported by the Naval Research (ONR) grant N00014-23-1-2737.

\newpage
\appendix 
\section{Additional Proofs for Randomized Pauli Mixedness Testing}\label{sec:ub-proof}
\subsection{Pauli Work Investment Lemma} 
We will prove that there exists a Levin's Work Investment style strategy for Pauli basis measurements.
\begin{proof}[Proof of~\cref{lem:pauli_work_invest}]
Before we proceed, we will first re-prove a classical Levin's Work Investment lemma.
\begin{lemma}[Levin's Work Investment Strategy]\label{lem:work_invest}
    Let $X \sim \unif[k]$. If $\expectDistrOf{}{f(X)} \geq \mu$ and $f(X)\geq 0$. Then, there exists $i \in [\lceil \log_3 k\rceil]$ such that,
    \begin{align*}
        \Pr[f(X) \geq 3^{i-2} \mu] \geq \frac{3^{-i}}{\lceil \log_3 k \rceil},
    \end{align*}
\end{lemma}
\begin{proof}
Let $M=\lceil \log_3 k \rceil$.
    We proceed by contradiction. Assume $\Pr[f(X) \geq 3^{i-2} \mu] < \frac{3^{-i}}{M}$ for all $i \in [M]$. We observe that $\Pr[f(X) \geq k \mu/3] = 0$ since any element on the support of $f(X)$ should have probability $\geq 1/k$. Thus, we can split the expectation up into a finite number of bins $B_j = \{x: 3^j  \mu \leq f(x) \leq 3^{j+1} \mu\}$.
    \begin{align*}
        \expectDistrOf{}{f(X)} &\leq \frac{\mu}{3} +  \sum_{j=1}^{M} \Pr[f(X) \in B_j] \cdot 3^{j+1} \mu \\
                            &< \frac{\mu}{3} + \frac{\mu}{3} \sum_{j=1}^{\lceil \log_3 k \rceil} \frac{1}{M} \leq \frac{\mu}{3} + \frac{\mu}{3} < \mu
    \end{align*}
    We have reached a contradiction.
\end{proof}
\noindent We can now prove \cref{lem:pauli_work_invest} as a direct consequence of \cref{lem:work_invest}.
    We look at the expectation of $L_P$,
    \begin{align*}
    \expectDistrOf{}{L_P} &= \frac{1}{3^N} \sum_{P \in \{X,Y,Z\}^{\otimes N}} \sum_{Q \triangleleft P} 3^{w(Q)} \alpha_Q^2 \\
    &= \frac{1}{3^N} \sum_{Q \in \mathcal{Q}} \alpha_Q^2 \cdot 3^{N-w(Q)} \cdot 3^{w(Q)} = L.
    \end{align*}
    The statement is proven by applying \cref{lem:work_invest} for the random variable $L_P$.
\end{proof}
\subsection{Correlation-Concentrated Testing}
We will prove the correctness of the \verb|KrawtCheck| subroutine, giving us a polynomial-time uniformity tester of correlation-sparse distributions.
\begin{proof}[Proof of~\cref{thm:sparse-test}]
    It is clear that \cref{alg:sparse} uses $\mathcal{O}(Mn) = \bigO{\frac{\sqrt{\binom{N}{k}}}{\beta\Delta}}$ samples. 

    \noindent We will show that $W_k$ is an unbiased statistic for $\|\alpha^{=k}\|_2^2$ by making use of the observation from~\cref{eq:k-sym}
    \begin{align*}
    f(x) = \sum_{|S| = k} \parity{x} = K_k^{(N)}(w_{H}(x)).
    \end{align*}
    Coupled with the fact that $d_{H}(x,y) = w_H(x \oplus y)$,
    \begin{align*}
        \expectDistrOf{}{W_k} &= \expectDistrOf{X,Y \sim p}{\krawt{k}{w_{H}(X \oplus Y)}} = \sum_{x} p(x) \sum_{y} p(y) f(x \oplus y)\\
        &= 2^N \sum_{x} p(x) (p*f)(x) = 4^N \dotprodB{p}{p*f}.
    \end{align*}
    By the convolution property of Boolean functions, $p*f (x) = \frac{1}{2^N} \sum_{|S|=k} \parity{x} \alpha(S)$. By Parseval's Theorem, the inner product $\dotprod{p}{p*f}_B$ is the inner product of their Fourier expansions, resulting in $\frac{1}{4^N} \|\alpha^{=k}\|_2^2$. 
    
    We have shown that $W_k$ is an unbiased statistic of $\|\alpha^{=k}\|_2^2$, which is $0$ when $p=u$ and $\geq \beta\Delta$ when $p \in \mathcal{P}$. Thus, it suffices to estimate $W_k$ with accuracy $\kappa_1 \cdot \beta \Delta$ when $p=u$ and $\kappa_2 \cdot \|\alpha^{=k}\|_2^2$ when $p \in \mathcal{P}$ to distinguish between the two classes. We proceed to bound the variance of $W_k$ to show that the error is well-controlled for detection,
    \begin{align*}
        \Var(W_k) &= \frac{1}{n^4} \expectDistrOf{}{\left(\sum_{i=1}^n \sum_{j=1}^n \krawt{k}{d_{H}(X_i,Y_i)}\right)^2} - \|\alpha^{=k}\|_2^4.
    \end{align*}
    When we expand the second moment, we will look at correlations between pairs $(X,Y), (X',Y')$. We will split the correlations into three cases: (1) when the pairs are disjoint, (2) when the pairs share one element, and (3) when the pairs are the same. 
    \begin{align*}
    \expectDistrOf{}{\left(\sum_{i=1}^n \sum_{j=1}^n \krawt{k}{d_{H}(X_i,Y_i)}\right)^2} =\; &n^2(n-1)^2 \|\alpha^{=k}\|_2^4 + 2n^3 \expectDistrOf{}{\krawt{k}{d_{H}(X,Y)}\krawt{k}{d_{H}(X,Y')}} \\
    &+ n^2 \expectDistrOf{}{\krawt{k}{d_{H}(X,Y)}^2}.
    \end{align*}
    The first case is simple as the two variables will be independent. Consider the expectation when the pairs of samples share exactly one element,
    \begin{align*}
        \expectDistrOf{}{\krawt{k}{d_{H}(X,Y)}\krawt{k}{d_{H}(X,Y')}} &= \sum_{x} p(x) \left(\sum_{y} p(y) \krawt{k}{w_{H}(x\oplus y)}\right)^2 \\
        &= 8^N \dotprodB{p}{(p*f)^2}
    \end{align*}
    We make the observation that 
    \begin{align*}
    (p*f)^2(x) = \frac{1}{4^N} \sum_{|S_1|=|S_2|=k} \alpha(S_1) \alpha(S_2) \cdot \chi_{S_1 \Delta S_2}(x). 
    \end{align*}
    Define a new set of coefficients that includes the empty set: $\tilde{\alpha}(S) = \alpha(S)$ and $\tilde{\alpha}(\emptyset) = 1$. Applying Parseval's Theorem,
    \begin{align*}
        \expectDistrOf{}{\krawt{k}{d_{H}(X,Y)}\krawt{k}{d_{H}(X,Y')}} = \sum_{|S_1|=|S_2|=k} \tilde{\alpha}(S_1 \Delta S_2) \alpha(S_1) \alpha(S_2).
    \end{align*}
    For convenience, we define the quantity 
    \begin{align*}
    g(k) \eqdef \begin{cases}
        2k & \text{if }  k \leq \nqubits/2 \\
        2(\nqubits-k) &  \text{otherwise}.
    \end{cases}
    \end{align*}
    We can now complete the bound with Cauchy-Schwarz inequality,
    \begin{align*}
        \expectDistrOf{}{\krawt{k}{d_{H}(X,Y)}\krawt{k}{d_{H}(X,Y')}} &\leq \sqrt{\sum_{|S_1|=|S_2|=k} \tilde{\alpha}(S_1 \Delta S_2)^2} \cdot \sqrt{\sum_{|S_1|=|S_2|=k} \alpha(S_1)^2 \alpha(S_2)^2} \\
        &\leq \sqrt{\sum_{\substack{|S| \leq g(k) \\ |S| \text{ even}}} 2 \binom{|S|}{|S|/2} \binom{N-|S|}{k-|S|/2} \cdot \tilde{\alpha}(S)^2}  \cdot \|\alpha^{=k}\|_2^2 \\
        &\leq \sqrt{ 2 \binom{N}{k}} \|\alpha\|_2 \|\alpha^{=k}\|_2^2 + \sqrt{2 \binom{N}{k}} \|\alpha^{=k}\|_2^2 \\
        &\leq \sqrt{2 \beta^{-1} \binom{N}{k}} \|\alpha^{=k}\|_2^3 + \sqrt{2 \binom{N}{k}} \|\alpha^{=k}\|_2^2,
    \end{align*}
    where we used the fact that $\sqrt{a+b} \leq \sqrt{a} + \sqrt{b}$ and the assumption $\|\alpha^{=k}\|_2^2 \geq \beta\|\alpha\|_2^2$ when $p \in \mathcal{P}^{(\Delta, k ,\beta)}$. Now, we shift our attention to bounding the second moment,
    \begin{align*}
        \expectDistrOf{}{\krawt{k}{d_{H}(X,Y)}^2} &= \sum_{x} p(x) \sum_{y} p(y) \krawt{k}{w_{H}(x \oplus y)}^2 \\
        &= 4^N \dotprodB{p}{p * f^2}.
    \end{align*}
    We make an observation regarding the convolution $p*f^2$.
    \begin{align*}
    f^2(x) &= \sum_{|S_1|=|S_2|=k} \chi_{S_1 \Delta S_2}(x) = 2 \sum_{\substack{1 \leq |S| \leq g(k) \\ |S | \text{ even}}} \binom{|S|}{|S|/2} \binom{N-|S|}{k-|S|/2} \chi_{S}(x) + \binom{N}{k} \chi_{\emptyset}(x) \\ 
    p*f^2(x) &= \frac{2}{2^N} \sum_{\substack{|S| \leq g(k) \\ |S | \text{ even}}} \binom{|S|}{|S|/2}  \binom{N-|S|}{k-|S|/2} \alpha(S) \cdot \chi_{S}(x) + \frac{\binom{N}{k}}{2^N}
    \end{align*}
    We can apply Parseval's Theorem once more to obtain the following second moment bound,
    \begin{align*}
        \expectDistrOf{}{\krawt{k}{d_{H}(X,Y)}^2} &= 2 \sum_{\substack{|S| \leq 2k \\ |S | \text{ even}}} \binom{|S|}{|S|/2} \binom{N-|S|}{k-|S|/2} \alpha(S)^2 + \binom{N}{k} \\
        &\leq 2 \binom{N}{k} \|\alpha\|_2^2 + 2 \binom{N}{k} \leq 2\beta^{-1}\binom{N}{k} \|\alpha^{=k}\|_2^2 + 2 \binom{N}{k}.
    \end{align*}
    Thus, 
    \begin{align*}
        \Var(W_k) &\leq \frac{ 2\sqrt{ 2 \binom{N}{k}} \|\alpha^{=k}\|_2^3}{\sqrt{\beta} n} + \frac{2\sqrt{2 \binom{N}{k}} \|\alpha^{=k}\|_2^2}{n} +\frac{2 \binom{N}{k} \|\alpha^{=k}\|_2^2}{\beta n^2} + \frac{2 \binom{N}{k}}{n^2}.
    \end{align*}
    With $n = \frac{400\sqrt{2 \binom{N}{k}}}{\beta \Delta}$, the variance under $p \in \mathcal{P}$ will be
    \begin{align*}
        \Var(W_k) \leq \frac{\sqrt{\beta \Delta} \|\alpha^{=k}\|_2^3}{200}  + \frac{\beta \Delta \|\alpha^{=k}\|_2^2}{200} + \frac{\beta \Delta \|\alpha^{=k}\|_2^2}{40000} + \frac{\beta^2 \Delta^2}{40000} \leq \frac{\|\alpha^{=k}\|_2^4}{49}.
    \end{align*}
    Since $\|\alpha^{=k}\| = 0$ when $p=u$, the variance under the null hypothesis will be
    \begin{align*}
    \Var(W_k) &\leq\frac{2\binom{N}{k}}{n^2} \leq \frac{\beta^2 \Delta^2}{40000}. 
    \end{align*}
    By Chebyshev's inequality,
    \begin{align*}
        \Pr_{p \in \mathcal{P}}\left[W_k \geq \frac{\|\alpha^{=k}\|_2^2}{50} \geq \frac{\beta \Delta}{50}\right] \geq \frac{3}{4}, \quad \Pr_{p=u}\left[W_k < \frac{\beta \Delta}{50}\right] \geq \frac{3}{4}.
    \end{align*}
    Thus, we have proved the correctness of~\cref{alg:sparse}.     \noindent With black box amplification, we have proved~\cref{thm:sparse-test}.
    \end{proof}
    \section{Covariance Bound for Deterministic Pauli Mixedness Testing}\label{sec:ub-det-proof}
    We will prove that the covariance of the Pauli coefficients are well-controlled for detecting the $\ell_2$ norm.
    \begin{proof}[Proof of~\cref{lem:cov-pauli}] Before simplifying the expression, we define the set-theoretic operations between Pauli observables that belong to the same measurement group,
    \begin{align*}
    Q \Delta Q' &\eqdef \bigotimes_{i=1}^\nqubits \indic{\sigma_i = \sigma'_i} \eye + \indic{\sigma_i \neq \sigma'_i = \eye} \sigma_i + \indic{\sigma'_i \neq \sigma_i = \eye} \sigma'_i, \\
    Q \cup Q' &\eqdef \bigotimes_{i=1}^\nqubits \indic{\sigma_i = \sigma'_i} \sigma_i + \indic{\sigma_i \neq \sigma'_i = \eye} \sigma_i + \indic{\sigma'_i \neq \sigma_i = \eye} \sigma'_i, \\
    Q \cap Q' &\eqdef \bigotimes_{i=1}^\nqubits \indic{\sigma_i = \sigma'_i} \sigma_i + \indic{\sigma_i \neq \sigma'_i} \eye,
    \end{align*}
    where $Q = \bigotimes_{i=1}^{\nqubits} \sigma_i$ and $Q' = \bigotimes_{i=1}^{\nqubits} \sigma'_i$. Note, the case $\eye \neq \sigma_i \neq \sigma'_i \neq \eye$ does not exist as there wouldn't exist a basis that covers all of the non-identity entries of both $Q$ and $Q'$ simultaneously. For the rest of the proof, we assume that the pair $Q,Q'$ will agree on non-identity tensor entries. We notice there are $3^{\nqubits - w(Q \cup Q')}$ unique occurrences of the pair $Q,Q'$ across all the measurement groups. We will now simplify the expression given in the lemma,
    \begin{align*}
        \sum_{P \in \{X,Y,Z\}^{\otimes \nqubits}} \sum_{\emptyset \subset S,S' \subseteq [\nqubits]} \frac{\alpha_P(S \Delta S')^2}{3^{\nqubits - |S|} \cdot 3^{\nqubits - |S'|}} &=  \sum_{Q, Q' \in \mathcal{Q}} \frac{\alpha_{Q \Delta Q'}^2}{3^{2N-w(Q)-w(Q')}} \cdot 3^{\nqubits - w(Q \cup Q')} \\
        &= \sum_{Q, Q' \in \mathcal{Q}} \frac{\alpha_{Q \Delta Q'}^2}{3^{N-(w(Q)+w(Q')- w(Q \cup Q'))}} \\
        &= \sum_{Q, Q' \in \mathcal{Q}} \frac{\alpha_{Q \Delta Q'}^2}{3^{N-w(Q \cap Q')}} 
    \end{align*}
    We now make the following observation. $Q$ and $Q'$ should be disjoint in the non-identity positions of $Q \Delta Q'$, and there are $2^{w(Q \Delta Q')}$ ways to perform the split. Furthermore, $Q$ and $Q'$ should be identical across the identity coordinates of $Q \Delta Q'$. Thus,
    \begin{align*}
        \sum_{Q, Q' \in \mathcal{Q}} \frac{\alpha_{Q \Delta Q'}^2}{3^{N-w(Q \cap Q')}}  &=  \frac{1}{3^\nqubits} \sum_{Z \in  \mathcal{Q} \cup \{\eye_d\}} \alpha_{Z}^2 \cdot 2^{w(Z)} \sum_{w=0}^{\nqubits - w(Z)} 3^{w} \cdot 3^{w} \binom{\nqubits-w(Z)}{w} \\
        &= \left(\frac{10}{3}\right)^{\nqubits} \sum_{z \in \mathcal{Q} \cup \{\eye_d\}} \frac{\alpha_Z^2}{5^{w(Z)}} = \left(\frac{10}{3}\right)^{\nqubits} + \left(\frac{10}{3}\right)^{\nqubits} \sum_{Q \in \mathcal{Q}} \frac{\alpha_Q^2}{5^{w(Q)}}.
    \end{align*}
    We have proven~\cref{lem:cov-pauli}.
    \end{proof}
    \section{Computational Efficiency of\texorpdfstring{~\cref{alg:sparse}}{Algorithm 2}}
    We now prove that correlation-concentrated uniformity testing can be done in $poly(n,\nqubits,k)$ time.
    \begin{proof}[Proof of~\cref{cor:eff-sparse}]
    An important observation is that the Krawtchouk polynomials can be generated through a 3-term recurrence~\cite[Definition 28]{vanDam2002QuantumComputationTheory},
    \begin{align*}
        \krawt{k}{t} = \frac{\nqubits -2t}{k} \krawt{k-1}{t} + \frac{\nqubits-k+2}{k} \krawt{k-2}{t}.
    \end{align*}    
    Since there is a constant number of additions, divisions and multiplications of integers in $[-2^N, 2^N]$, each recurrence can be computed in time $\bigO{\nqubits \log \nqubits}$ using standard fast multiplication algorithms~\cite{harvey2021integer}, as division can be reduced to multiplication via Newton-Raphson method. Therefore, iteratively computing $\krawt{k}{t}$ can be done in $\bigO{k\nqubits \log \nqubits}$ time with standard dynamic programming methods. Coupled with the additional $\bigO{n^2 \cdot \nqubits}$ time for computing the Hamming distance between samples and summing over $n^2$ integers of magnitude $2^\nqubits$, we have an algorithm that runs in $\bigO{n^2 \cdot k\nqubits \log \nqubits + 2 n^2 \cdot \nqubits} = \bigO{n^2 \cdot k\nqubits \log \nqubits}$ time, yielding an exponential improvement over the naïve parity-enumeration algorithm.
    \end{proof}

% that's all folks
\bibliography{ref}

@article{flamia2011direct,
  title = {Direct Fidelity Estimation from Few Pauli Measurements},
  author = {Flammia, Steven T. and Liu, Yi-Kai},
  journal = {Phys. Rev. Lett.},
  volume = {106},
  issue = {23},
  pages = {230501},
  numpages = {4},
  year = {2011},
  month = {Jun},
  publisher = {American Physical Society},
  doi = {10.1103/PhysRevLett.106.230501},
  url = {https://link.aps.org/doi/10.1103/PhysRevLett.106.230501}
}

@article{10.1145/3717450,
author = {Nayak, Ashwin and Lowe, Angus},
title = {Lower Bounds for Learning Quantum States with Single-Copy Measurements},
year = {2025},
issue_date = {March 2025},
publisher = {Association for Computing Machinery},
address = {New York, NY, USA},
volume = {17},
number = {1},
issn = {1942-3454},
url = {https://doi.org/10.1145/3717450},
doi = {10.1145/3717450},
journal = {ACM Trans. Comput. Theory},
month = mar,
articleno = {7},
numpages = {42}
}

@article{Elben_2022,
   title={The randomized measurement toolbox},
   volume={5},
   ISSN={2522-5820},
   url={http://dx.doi.org/10.1038/s42254-022-00535-2},
   DOI={10.1038/s42254-022-00535-2},
   number={1},
   journal={Nature Reviews Physics},
   publisher={Springer Science and Business Media LLC},
   author={Elben, Andreas and Flammia, Steven T. and Huang, Hsin-Yuan and Kueng, Richard and Preskill, John and Vermersch, Benoît and Zoller, Peter},
   year={2022},
   month=dec, pages={9–24} }

@article{Yuen_2023,
   title={An Improved Sample Complexity Lower Bound for (Fidelity) Quantum State Tomography},
   volume={7},
   ISSN={2521-327X},
   url={http://dx.doi.org/10.22331/q-2023-01-03-890},
   DOI={10.22331/q-2023-01-03-890},
   journal={Quantum},
   publisher={Verein zur Forderung des Open Access Publizierens in den Quantenwissenschaften},
   author={Yuen, Henry},
   year={2023},
   month=jan, pages={890} }

@INPROCEEDINGS{BOW17,
  Title                    = {Quantum state certification},
  Author                   = {C. B\u{a}descu and R. O'Donnell and J. Wright},
Booktitle                = {Proceedings of the Forty-Nineth Annual ACM on Symposium on Theory of Computing},
  Series                   = {STOC '19},
 year = {2019}
}

@Article{MdW13,
  Title                    = {A Survey of Quantum Property Testing},

  Author                   = {A. Montanaro and R. de Wolf},
  Year                     = {2016},
Volume                     ={7},
  Journal                   = {Theory of Computing Graduate Surveys},
}

@InProceedings{OW15,
  Title                    = {Quantum Spectrum Testing},
  Author                   = {O'Donnell, R. and Wright, J.},
  Booktitle                = {Proceedings of the Forty-Seventh Annual ACM on Symposium on Theory of Computing},
  Year                     = {2015},
  Pages                    = {529--538},
  Series                   = {STOC '15},
}

@Article{GrossLFBE10,
  title={Quantum State Tomography via Compressed Sensing},
   volume={105},
   ISSN={1079-7114},
   url={http://dx.doi.org/10.1103/PhysRevLett.105.150401},
   DOI={10.1103/physrevlett.105.150401},
   number={15},
   journal={Physical Review Letters},
   publisher={American Physical Society (APS)},
   author={Gross, David and Liu, Yi-Kai and Flammia, Steven T. and Becker, Stephen and Eisert, Jens},
   year={2010},
   month=oct }

@article{AGKE15,
  title = {Reliable quantum certification of photonic state preparations},
  author = {Leandro Aolita and Christian Gogolin and Martin Kliesch and Jens Eisert},
  journal = {Nature Communications},
  volume = {6},
  issue = {8498},
  year = {2015},
}

@article{PhysRevLett.107.210404,
  title = {Practical Characterization of Quantum Devices without Tomography},
  author = {da Silva, Marcus P. and Landon-Cardinal, Olivier and Poulin, David},
  journal = {Phys. Rev. Lett.},
  volume = {107},
  issue = {21},
  pages = {210404},
  year = {2011},
}

@inproceedings{Chen0L24memory,
  author       = {Sitan Chen and
                  Jerry Li and
                  Allen Liu},
  editor       = {Bojan Mohar and
                  Igor Shinkar and
                  Ryan O'Donnell},
  title        = {An Optimal Tradeoff between Entanglement and Copy Complexity for State
                  Tomography},
  booktitle    = {Proceedings of the 56th Annual {ACM} Symposium on Theory of Computing,
                  {STOC} 2024, Vancouver, BC, Canada, June 24-28, 2024},
  pages        = {1331--1342},
  publisher    = {{ACM}},
  year         = {2024},
  url          = {https://doi.org/10.1145/3618260.3649704},
  doi          = {10.1145/3618260.3649704},
  bibsource    = {dblp computer science bibliography, https://dblp.org}
}

@article{Flammia_2012,
   title={Quantum tomography via compressed sensing: error bounds, sample complexity and efficient estimators},
   volume={14},
   ISSN={1367-2630},
   url={http://dx.doi.org/10.1088/1367-2630/14/9/095022},
   DOI={10.1088/1367-2630/14/9/095022},
   number={9},
   journal={New Journal of Physics},
   publisher={IOP Publishing},
   author={Flammia, Steven T and Gross, David and Liu, Yi-Kai and Eisert, Jens},
   year={2012},
   month={Sep},
   pages={095022}
}

@inproceedings{BubeckC020,
  author       = {S{\'{e}}bastien Bubeck and
                  Sitan Chen and
                  Jerry Li},
  editor       = {Sandy Irani},
  title        = {Entanglement is Necessary for Optimal Quantum Property Testing},
  booktitle    = {61st {IEEE} Annual Symposium on Foundations of Computer Science, {FOCS}
                  2020, Durham, NC, USA, November 16-19, 2020},
  pages        = {692--703},
  publisher    = {{IEEE}},
  year         = {2020},
  url          = {https://doi.org/10.1109/FOCS46700.2020.00070},
  doi          = {10.1109/FOCS46700.2020.00070},
  bibsource    = {dblp computer science bibliography, https://dblp.org}
}

@inproceedings{Chen0HL22,
  author       = {Sitan Chen and
                  Jerry Li and
                  Brice Huang and
                  Allen Liu},
  title        = {Tight Bounds for Quantum State Certification with Incoherent Measurements},
  booktitle    = {63rd {IEEE} Annual Symposium on Foundations of Computer Science, {FOCS}
                  2022, Denver, CO, USA, October 31 - November 3, 2022},
  pages        = {1205--1213},
  publisher    = {{IEEE}},
  year         = {2022},
  url          = {https://doi.org/10.1109/FOCS54457.2022.00118},
  doi          = {10.1109/FOCS54457.2022.00118},
  bibsource    = {dblp computer science bibliography, https://dblp.org}
}

@article{DiakonikolasGPP19,
  author       = {Ilias Diakonikolas and
                  Themis Gouleakis and
                  John Peebles and
                  Eric Price},
  title        = {Collision-Based Testers are Optimal for Uniformity and Closeness},
  journal      = {Chic. J. Theor. Comput. Sci.},
  volume       = {2019},
  year         = {2019},
  url          = {http://cjtcs.cs.uchicago.edu/articles/2019/1/contents.html},
  bibsource    = {dblp computer science bibliography, https://dblp.org}
}

@article{canonne2022topics,
	url = {http://dx.doi.org/10.1561/0100000114},
	year = {2022},
	volume = {19},
	journal = {Foundations and Trends® in Communications and Information Theory},
	title = {Topics and Techniques in Distribution Testing: A Biased but Representative Sample},
	doi = {10.1561/0100000114},
	issn = {1567-2190},
	number = {6},
	pages = {1032-1198},
	author = {Clément L. Canonne}
}

@inproceedings{ChenLO22instance,
  author       = {Sitan Chen and
                  Jerry Li and
                  Ryan O'Donnell},
  editor       = {Po{-}Ling Loh and
                  Maxim Raginsky},
  title        = {Toward Instance-Optimal State Certification With Incoherent Measurements},
  booktitle    = {Conference on Learning Theory, 2-5 July 2022, London, {UK}},
  series       = {Proceedings of Machine Learning Research},
  volume       = {178},
  pages        = {2541--2596},
  publisher    = {{PMLR}},
  year         = {2022},
  url          = {https://proceedings.mlr.press/v178/chen22b.html},
  bibsource    = {dblp computer science bibliography, https://dblp.org}
}

@article{flamian2023tomography,
  author       = {Steven T. Flammia and
                  Ryan O'Donnell},
  title        = {Quantum chi-squared tomography and mutual information testing},
   journal={Quantum},
  volume={8},
  pages={1381},
  year={2024},
  publisher={Verein zur F{\"o}rderung des Open Access Publizierens in den Quantenwissenschaften},
  url          = {https://doi.org/10.48550/arXiv.2305.18519},
  doi          = {10.48550/ARXIV.2305.18519},
}

@inproceedings{BadescuO019,
  author       = {Costin Badescu and
                  Ryan O'Donnell and
                  John Wright},
  editor       = {Moses Charikar and
                  Edith Cohen},
  title        = {Quantum state certification},
  booktitle    = {Proceedings of the 51st Annual {ACM} {SIGACT} Symposium on Theory
                  of Computing, {STOC} 2019, Phoenix, AZ, USA, June 23-26, 2019},
  pages        = {503--514},
  publisher    = {{ACM}},
  year         = {2019},
  url          = {https://doi.org/10.1145/3313276.3316344},
  doi          = {10.1145/3313276.3316344},
  bibsource    = {dblp computer science bibliography, https://dblp.org}
}

@inproceedings{chen2023does,
  author       = {Sitan Chen and
                  Brice Huang and
                  Jerry Li and
                  Allen Liu and
                  Mark Sellke},
  title        = {When Does Adaptivity Help for Quantum State Learning?},
  booktitle    = {64th {IEEE} Annual Symposium on Foundations of Computer Science, {FOCS}
                  2023, Santa Cruz, CA, USA, November 6-9, 2023},
  pages        = {391--404},
  publisher    = {{IEEE}},
  year         = {2023},
  url          = {https://doi.org/10.1109/FOCS57990.2023.00029},
  doi          = {10.1109/FOCS57990.2023.00029},
  bibsource    = {dblp computer science bibliography, https://dblp.org}
}

@article{HaahHJWY17,
  author       = {Jeongwan Haah and
                  Aram W. Harrow and
                  Zhengfeng Ji and
                  Xiaodi Wu and
                  Nengkun Yu},
  title        = {Sample-Optimal Tomography of Quantum States},
  journal      = {{IEEE} Trans. Inf. Theory},
  volume       = {63},
  number       = {9},
  pages        = {5628--5641},
  year         = {2017},
  url          = {https://doi.org/10.1109/TIT.2017.2719044},
  doi          = {10.1109/TIT.2017.2719044},
  bibsource    = {dblp computer science bibliography, https://dblp.org}
}

@article{KRT14,
  title={Low rank matrix recovery from rank one measurements},
  author={Richard Kueng and Holger Rauhut and Ulrich Terstiege},
  journal={Applied and Computational Harmonic Analysis},
  volume={42},
  number={1},
  pages={88--116},
  year={2017},
  publisher={Elsevier}
}

@inproceedings{ODonnellW17,
  author       = {Ryan O'Donnell and
                  John Wright},
  editor       = {Hamed Hatami and
                  Pierre McKenzie and
                  Valerie King},
  title        = {Efficient quantum tomography {II}},
  booktitle    = {Proceedings of the 49th Annual {ACM} {SIGACT} Symposium on Theory
                  of Computing, {STOC} 2017, Montreal, QC, Canada, June 19-23, 2017},
  pages        = {962--974},
  publisher    = {{ACM}},
  year         = {2017},
  url          = {https://doi.org/10.1145/3055399.3055454},
  doi          = {10.1145/3055399.3055454},
  bibsource    = {dblp computer science bibliography, https://dblp.org}
}

@inproceedings{ODonnellW16,
  author       = {Ryan O'Donnell and
                  John Wright},
  editor       = {Daniel Wichs and
                  Yishay Mansour},
  title        = {Efficient quantum tomography},
  booktitle    = {Proceedings of the 48th Annual {ACM} {SIGACT} Symposium on Theory
                  of Computing, {STOC} 2016, Cambridge, MA, USA, June 18-21, 2016},
  pages        = {899--912},
  publisher    = {{ACM}},
  year         = {2016},
  url          = {https://doi.org/10.1145/2897518.2897544},
  doi          = {10.1145/2897518.2897544},
  bibsource    = {dblp computer science bibliography, https://dblp.org}
}

@article{huang2020predicting,
  title={Predicting many properties of a quantum system from very few measurements},
  author={Huang, Hsin-Yuan and Kueng, Richard and Preskill, John},
  journal={Nature Physics},
  volume={16},
  number={10},
  pages={1050--1057},
  year={2020},
  publisher={Nature Publishing Group UK London}
}

@article{Aaronson20,
  author       = {Scott Aaronson},
  title        = {Shadow Tomography of Quantum States},
  journal      = {{SIAM} J. Comput.},
  volume       = {49},
  number       = {5},
  year         = {2020},
  url          = {https://doi.org/10.1137/18M120275X},
  doi          = {10.1137/18M120275X},
  bibsource    = {dblp computer science bibliography, https://dblp.org}
}

@inproceedings{BadescuO21,
  author       = {Costin Badescu and
                  Ryan O'Donnell},
  editor       = {Samir Khuller and
                  Virginia Vassilevska Williams},
  title        = {Improved Quantum data analysis},
  booktitle    = {{STOC} '21: 53rd Annual {ACM} {SIGACT} Symposium on Theory of Computing,
                  Virtual Event, Italy, June 21-25, 2021},
  pages        = {1398--1411},
  publisher    = {{ACM}},
  year         = {2021},
  url          = {https://doi.org/10.1145/3406325.3451109},
  doi          = {10.1145/3406325.3451109},
  bibsource    = {dblp computer science bibliography, https://dblp.org}
}

@ARTICLE{Yu2023almost,
  author={Yu, Nengkun},
  journal={IEEE Transactions on Information Theory}, 
  title={Almost Tight Sample Complexity Analysis of Quantum Identity Testing by Pauli Measurements}, 
  year={2023},
  volume={69},
  number={8},
  pages={5060-5068},
  doi={10.1109/TIT.2023.3271206}}

@inproceedings{Yu21sample,
  author       = {Nengkun Yu},
  editor       = {James R. Lee},
  title        = {Sample Efficient Identity Testing and Independence Testing of Quantum
                  States},
  booktitle    = {12th Innovations in Theoretical Computer Science Conference, {ITCS}
                  2021, January 6-8, 2021, Virtual Conference},
  series       = {LIPIcs},
  volume       = {185},
  pages        = {11:1--11:20},
  publisher    = {Schloss Dagstuhl - Leibniz-Zentrum f{\"{u}}r Informatik},
  year         = {2021},
  url          = {https://doi.org/10.4230/LIPIcs.ITCS.2021.11},
  doi          = {10.4230/LIPICS.ITCS.2021.11},
  bibsource    = {dblp computer science bibliography, https://dblp.org}
}

@article{Yu2020Pauli,
      title={Sample efficient tomography via {P}auli Measurements}, 
      author={Nengkun Yu},
      journal={CoRR},
      year={2020},
      volume={abs/2009.04610},
      eprint={2009.04610},
      archivePrefix={arXiv},
      primaryClass={quant-ph},
      url={https://arxiv.org/abs/2009.04610}, 
}

@inproceedings{Liu2011universal,
  author       = {Yi{-}Kai Liu},
  editor       = {John Shawe{-}Taylor and
                  Richard S. Zemel and
                  Peter L. Bartlett and
                  Fernando C. N. Pereira and
                  Kilian Q. Weinberger},
  title        = {Universal low-rank matrix recovery from Pauli measurements},
  booktitle    = {Advances in Neural Information Processing Systems 24: 25th Annual
                  Conference on Neural Information Processing Systems 2011. Proceedings
                  of a meeting held 12-14 December 2011, Granada, Spain},
  pages        = {1638--1646},
  year         = {2011},
  url          = {https://proceedings.neurips.cc/paper/2011/hash/e820a45f1dfc7b95282d10b6087e11c0-Abstract.html},
  bibsource    = {dblp computer science bibliography, https://dblp.org}
}

@article{Paninski08,
  author       = {Liam Paninski},
  title        = {A Coincidence-Based Test for Uniformity Given Very Sparsely Sampled
                  Discrete Data},
  journal      = {{IEEE} Trans. Inf. Theory},
  volume       = {54},
  number       = {10},
  pages        = {4750--4755},
  year         = {2008},
  url          = {https://doi.org/10.1109/TIT.2008.928987},
  doi          = {10.1109/TIT.2008.928987},
  bibsource    = {dblp computer science bibliography, https://dblp.org}
}

@ARTICLE{ACLST22iiuic,
  author={Acharya, Jayadev and Canonne, Clément L. and Liu, Yuhan and Sun, Ziteng and Tyagi, Himanshu},
  journal={IEEE Transactions on Information Theory}, 
  title={Interactive Inference Under Information Constraints}, 
  year={2022},
  volume={68},
  number={1},
  pages={502-516},
  doi={10.1109/TIT.2021.3123905}}

@misc{grewal2026paulpure,
      title={Nearly Time-Optimal Pure State Tomography with Pauli Measurements}, 
      author={Sabee Grewal and Meghal Gupta and William He and Aniruddha Sen and Mihir Singhal},
      year={2026},
      eprint={2601.04444},
      archivePrefix={arXiv},
      primaryClass={quant-ph},
      url={https://arxiv.org/abs/2601.04444}, 
}

@InProceedings{Liu2024role,
  title = 	 {The role of randomness in quantum state certification with unentangled measurements},
  author =       {Liu, Yuhan and Acharya, Jayadev},
  booktitle = 	 {Proceedings of Thirty Seventh Conference on Learning Theory},
  pages = 	 {3523--3555},
  year = 	 {2024},
  editor = 	 {Agrawal, Shipra and Roth, Aaron},
  volume = 	 {247},
  series = 	 {Proceedings of Machine Learning Research},
  month = 	 {30 Jun--03 Jul},
  publisher =    {PMLR},
  url = 	 {https://proceedings.mlr.press/v247/liu24a.html}
}

@article{liu2024restricted,
      title={Quantum state testing with restricted measurements}, 
      author={Yuhan Liu and Jayadev Acharya},
      year={2024},
      journal={CoRR},
      volume= {abs/2408.17439},
      eprint={2408.17439},
      archivePrefix={arXiv},
      primaryClass={quant-ph},
      url={https://arxiv.org/abs/2408.17439}, 
}

@inproceedings{ADLY2025Paulinot,
author = {Acharya, Jayadev and Dharmavarapu, Abhilash and Liu, Yuhan and Yu, Nengkun},
title = {Pauli Measurements Are Not Optimal for Single-Copy Tomography},
year = {2025},
isbn = {9798400715105},
publisher = {Association for Computing Machinery},
address = {New York, NY, USA},
url = {https://doi.org/10.1145/3717823.3718248},
doi = {10.1145/3717823.3718248},
booktitle = {Proceedings of the 57th Annual ACM Symposium on Theory of Computing},
pages = {718–729},
numpages = {12},
location = {Prague, Czechia},
series = {STOC '25}
}

@misc{scharnhorst2025,
      title={Optimal lower bounds for quantum state tomography}, 
      author={Thilo Scharnhorst and Jack Spilecki and John Wright},
      year={2025},
      eprint={2510.07699},
      archivePrefix={arXiv},
      primaryClass={quant-ph},
      url={https://arxiv.org/abs/2510.07699}, 
}

@misc{acharya2025single,
      title={Pauli Measurements Are Near-Optimal for Single-Qubit Tomography}, 
      author={Jayadev Acharya and Abhilash Dharmavarapu and Yuhan Liu and Nengkun Yu},
      year={2025},
      eprint={2507.22001},
      archivePrefix={arXiv},
      primaryClass={quant-ph},
      url={https://arxiv.org/abs/2507.22001}, 
}

@misc{wang2024sampleoptimal,
      title={Sample-Optimal Quantum Estimators for Pure-State Trace Distance and Fidelity via Samplizer}, 
      author={Qisheng Wang and Zhicheng Zhang},
      year={2024},
      eprint={2410.21201},
      archivePrefix={arXiv},
      primaryClass={quant-ph},
      url={https://arxiv.org/abs/2410.21201}, 
}

@misc{odonnell2025instanceocert,
      title={Instance-Optimal Quantum State Certification with Entangled Measurements}, 
      author={Ryan O'Donnell and Chirag Wadhwa},
      year={2025},
      eprint={2507.06010},
      archivePrefix={arXiv},
      primaryClass={quant-ph},
      url={https://arxiv.org/abs/2507.06010}, 
}

@InProceedings{odonnell2018closenesskwiseuniformity,
  author =	{O'Donnell, Ryan and Zhao, Yu},
  title =	{{On Closeness to k-Wise Uniformity}},
  booktitle =	{Approximation, Randomization, and Combinatorial Optimization. Algorithms and Techniques (APPROX/RANDOM 2018)},
  pages =	{54:1--54:19},
  series =	{Leibniz International Proceedings in Informatics (LIPIcs)},
  ISBN =	{978-3-95977-085-9},
  ISSN =	{1868-8969},
  year =	{2018},
  volume =	{116},
  editor =	{Blais, Eric and Jansen, Klaus and D. P. Rolim, Jos\'{e} and Steurer, David},
  publisher =	{Schloss Dagstuhl -- Leibniz-Zentrum f{\"u}r Informatik},
  address =	{Dagstuhl, Germany},
  URL =		{https://drops.dagstuhl.de/entities/document/10.4230/LIPIcs.APPROX-RANDOM.2018.54},
  URN =		{urn:nbn:de:0030-drops-94581},
  doi =		{10.4230/LIPIcs.APPROX-RwANDOM.2018.54}
}

@misc{odonnell2021analysisbooleanfunctions,
      title={Analysis of Boolean Functions}, 
      author={Ryan O'Donnell},
      year={2021},
      eprint={2105.10386},
      archivePrefix={arXiv},
      primaryClass={cs.DM},
      url={https://arxiv.org/abs/2105.10386}, 
}

@inbook{macwilliams1978theory,
  title={The Theory of Error-Correcting Codes},
  author={MacWilliams, F.J. and MacWilliams, F.J. and Sloane, N.J.A.},
  number={v. 16},
  chapter={10},
  isbn={9780444851932},
  lccn={76041296},
  series={Mathematical Studies},
  url={https://books.google.com/books?id=LuomAQAAIAAJ},
  year={1978},
  publisher={Elsevier Science}
}

@InProceedings{goldreich2014input,
  author =	{Goldreich, Oded},
  title =	{{On Multiple Input Problems in Property Testing}},
  booktitle =	{Approximation, Randomization, and Combinatorial Optimization. Algorithms and Techniques (APPROX/RANDOM 2014)},
  pages =	{704--720},
  series =	{Leibniz International Proceedings in Informatics (LIPIcs)},
  ISBN =	{978-3-939897-74-3},
  ISSN =	{1868-8969},
  year =	{2014},
  volume =	{28},
  editor =	{Jansen, Klaus and Rolim, Jos\'{e} and Devanur, Nikhil R. and Moore, Cristopher},
  publisher =	{Schloss Dagstuhl -- Leibniz-Zentrum f{\"u}r Informatik},
  address =	{Dagstuhl, Germany},
  URL =		{https://drops.dagstuhl.de/entities/document/10.4230/LIPIcs.APPROX-RANDOM.2014.704},
  URN =		{urn:nbn:de:0030-drops-47336},
  doi =		{10.4230/LIPIcs.APPROX-RANDOM.2014.704}
}

@inproceedings{levin1985oneway,
author = {Levin, L A},
title = {One-way functions and pseudorandom generators},
year = {1985},
isbn = {0897911512},
publisher = {Association for Computing Machinery},
address = {New York, NY, USA},
url = {https://doi.org/10.1145/22145.22185},
doi = {10.1145/22145.22185},
booktitle = {Proceedings of the Seventeenth Annual ACM Symposium on Theory of Computing},
pages = {363–365},
numpages = {3},
location = {Providence, Rhode Island, USA},
series = {STOC '85}
}

@article{goldreich2002testing,
  author    = {Goldreich, Oded and Ron, Dana},
  title     = {Testing connectivity of bounded-degree graphs},
  journal   = {Random Structures \& Algorithms},
  volume    = {21},
  number    = {1},
  pages     = {42--64},
  year      = {2002},
  publisher = {Wiley Online Library},
  doi       = {10.1002/rsa.10045}
}

@inproceedings{berman2014lp,
  author    = {Berman, Piotr and Raskhodnikova, Sofya and Yaroslavtsev, Grigory},
  title     = {$L_p$-testing},
  booktitle = {Proceedings of the Forty-Sixth Annual ACM Symposium on Theory of Computing (STOC)},
  pages     = {164--173},
  year      = {2014},
  doi       = {10.1145/2591796.2591823}
}

@misc{wadhwa2026optimal,
      title={Optimal Quantum State Testing Even with Limited Entanglement}, 
      author={Chirag Wadhwa and Sitan Chen},
      year={2026},
      eprint={2604.07460},
      archivePrefix={arXiv},
      primaryClass={quant-ph},
      url={https://arxiv.org/abs/2604.07460}, 
}

@article{harvey2021integer,
  title={Integer multiplication in time O(nlog$\backslash$,n)},
  author={Harvey, David and Van Der Hoeven, Joris},
  journal={Annals of Mathematics},
  volume={193},
  number={2},
  pages={563--617},
  year={2021},
  publisher={Department of Mathematics, Princeton University Princeton, New Jersey, USA}
}

@phdthesis{vanDam2002QuantumComputationTheory,
  author       = {Willem Klaas van Dam},
  title        = {On Quantum Computation Theory},
  school       = {University of Amsterdam},
  year         = {2002},
  address      = {Amsterdam, The Netherlands},
  type         = {Ph.D. thesis},
  isbn         = {978-90-5776-091-4},
  note         = {ILLC Dissertation Series, DS-2002-04},
  url          = {https://dare.uva.nl/id/57c47272-e0cc-4949-826a-4a6df34e71f9}
}
\bibliographystyle{alpha}
\end{document}